\documentclass[onecolumn,aps]{quantumarticle}
\usepackage{latexsym,amsmath}
\usepackage{bm}

\newcommand{\Sect}[1]{Section~\ref{#1}}
\newcommand{\sect}[1]{Sec.~\ref{#1}}
\newcommand{\Sects}[1]{Sections~\ref{#1}}

\newcommand{\eqn}[1]{Eq.~\eqref{#1}}

\newcommand{\eqns}[1]{Eqs.~\eqref{#1}}
\newcommand{\thm}[1]{Theorem~\ref{#1}}

\newcommand{\axm}[1]{Axiom~\ref{#1}}
\newcommand{\axms}[1]{Axioms~\ref{#1}}

\newcommand{\refcite}[1]{\cite{#1}}
\newcommand{\refcites}[1]{\cite{#1}}

\usepackage{amsthm}
\newtheorem{axiom}{Axiom}
\newtheorem*{axiom*}{Axiom}
\newtheorem{theorem}{Theorem}
\newtheorem*{theorem*}{Theorem}

\newtheorem*{lemma*}{Lemma}

\newtheorem*{corollary*}{Corollary}

\newtheorem*{proposition*}{Proposition}

\newtheorem*{definition*}{Definition}

\newcommand{\me}{\mathrm{e}}
\newcommand{\mi}{\mathrm{i}}

\newcommand{\T}{\mathrm{T}}

\newcommand{\Vol}{\mathcal{V}}

\newcommand{\dif}{\mathrm{d}}

\newcommand{\deriv}[2]{\frac{\dif #1}{\dif #2}}
\newcommand{\pderiv}[2]{\frac{\partial #1}{\partial #2}}

\newcommand{\dt}{\dif t}

\newcommand{\dl}{\dif \lambda}

\newcommand{\dthree}{\dif^{3}}
\newcommand{\dfour}{\dif^{4}}

\newcommand{\intthree}{\int \dthree}
\newcommand{\intfour}{\int \dfour}

\newcommand{\Dif}{\mathrm{D}}

\newcommand{\Dfour}{\Dif^{4}}

\newcommand{\intDfour}{\int \Dfour}

\newcommand{\threevec}[1]{\bm{#1}}

\newcommand{\threex}{\threevec{x}}
\newcommand{\threexp}{\threex'}
\newcommand{\threep}{\threevec{p}}
\newcommand{\threepp}{\threep'}

\newcommand{\E}[1]{\omega_{#1}} 
\newcommand{\Ep}{\E{\threep}}
\newcommand{\Epp}{\E{\threepp}}

\newcommand{\qmu}{q^{\mu}}
\newcommand{\qmul}{\qmu(\lambda)}

\newcommand{\qdot}{\dot{q}}

\newcommand{\qdotsq}{\qdot^{2}}

\newcommand{\xz}{x_{0}}

\newcommand{\threexz}{\threex_{0}}
\newcommand{\lambdaz}{\lambda_{0}}

\newcommand{\In}{_{\mathrm{in}}}
\newcommand{\Out}{_{\mathrm{out}}}
\newcommand{\I}{_{\mathrm{I}}}
\newcommand{\F}{_{\mathrm{F}}}

\newcommand{\Int}{^{\mathrm{int}}}

\newcommand{\op}[1]{\hat{#1}}
\newcommand{\opbar}[1]{\op{\bar{#1}}}

\newcommand{\opA}{\op{A}}
\newcommand{\opB}{\op{B}}

\newcommand{\opH}{\op{H}}
\newcommand{\opG}{\op{G}}

\newcommand{\opP}{\op{P}}

\newcommand{\opS}{\op{S}}
\newcommand{\opU}{\op{U}}
\newcommand{\opV}{\op{V}}
\newcommand{\opX}{\op{X}}

\newcommand{\opv}{\op{v}}

\newcommand{\conj}{^{*}}

\newcommand{\adj}{^{\dag}}
\newcommand{\dadj}{^{\ddag}}

\newcommand{\adv}[1]{#1^{+}}       

\newcommand{\param}[5]{#1#2#3;#4#5}
\newcommand{\paraml}[2]{\param{#1}({#2}{\lambda})}
\newcommand{\paramlz}[2]{\param{#1}({#2}{\lambdaz})}

\newcommand{\psil}[1]{\paraml{\psi}{#1}}
\newcommand{\psilz}[1]{\paramlz{\psi}{#1}}

\newcommand{\psixl}{\psil{x}}
\newcommand{\psixlz}{\psilz{\xz}}

\newcommand{\paramk}[5]{\param{#1}{#2}{#3}{#4}{#5}}

\newcommand{\kersym}{\Delta}

\newcommand{\kerneld}{\paramk{\kersym}({x - \xz}{\lambda - \lambdaz})}

\newcommand{\kersymn}{\Delta_{n}}
\newcommand{\kerneln}{\paramk{\kersymn}({x - \xz}{\lambda - \lambdaz})}

\newcommand{\kernelnar}{\paramk{\kersymn}({x - \xz}{\pm (\lambda - \lambdaz)})}

\newcommand{\propsym}{\Delta}
\newcommand{\prop}{\propsym(x - \xz)}

\newcommand{\propasym}{\adv{\propsym}}
\newcommand{\propa}{\propasym(x - \xz)}

\newcommand{\propsymn}{\propsym_{n}}

\newcommand{\propasymn}{\adv{\propsymn}}
\newcommand{\propan}{\propasymn(x - \xz)}

\newcommand{\fnW}{\mathcal{W}}

\newcommand{\GroupP}{\mathcal{P}}

\newcommand{\GroupSO}{\mathrm{SO}}

\newcommand{\Hilb}[1]{\mathcal{#1}}

\newcommand{\HilbA}{\Hilb{A}}
\newcommand{\HilbD}{\Hilb{D}}

\newcommand{\HilbH}{\Hilb{H}}
\newcommand{\HilbK}{\Hilb{K}}
\newcommand{\HilbL}{\Hilb{L}}

\newcommand{\bra}[1]{\langle #1 |}
\newcommand{\ket}[1]{| #1 \rangle}
\newcommand{\inner}[2]{\langle #1 | #2 \rangle}

\newcommand{\params}[2]{\param{}{}{#1}{#2}{}}

\newcommand{\vac}{0}
\newcommand{\bravac}{\bra{\vac}}
\newcommand{\ketvac}{\ket{\vac}}

\newcommand{\bral}[1]{\bra{\params{#1}{\lambda}}}
\newcommand{\bralz}[1]{\bra{\params{#1}{\lambdaz}}}
\newcommand{\bralp}[1]{\bra{\params{#1}{\lambda'}}}

\newcommand{\brax}{\bra{x}}

\newcommand{\ketl}[1]{\ket{\params{#1}{\lambda}}}
\newcommand{\ketlz}[1]{\ket{\params{#1}{\lambdaz}}}

\newcommand{\ketxl}{\ketl{x}}

\newcommand{\ketx}{\ket{x}}

\newcommand{\ketp}{\ket{p}}

\newcommand{\innerll}[2]{\langle #1; \lambda | #2; \lambda \rangle}
\newcommand{\innerllz}[2]{\langle #1; \lambda | #2; \lambdaz \rangle}

\newcommand{\argmn}[3]{#1_{#2}, \ldots, #1_{#3}}
\newcommand{\argn}[2]{\argmn{#1}{1}{#2}}
\newcommand{\argN}[1]{\argn{#1}{N}}
\newcommand{\argNx}{\argN{x}}

\newcommand{\seqmn}[3]{#1{#2}; \ldots ;#1{#3}}
\newcommand{\seqn}[2]{\seqmn{#1}{1}{#2}}
\newcommand{\seqN}[1]{\seqn{#1}{N}}

\newcommand{\xli}[1]{x_{#1}, \lambda_{#1}}

\newcommand{\xlzi}[1]{x_{#1}, \lambdaz}

\newcommand{\xliN}{\seqN{\xli}}
\newcommand{\xlziN}{\argNx; \lambdaz}

\newcommand{\na}{n_{+}}
\newcommand{\nr}{n_{-}}
\newcommand{\nar}{n_{\pm}}

\newcommand{\oppsi}{\op{\psi}}

\newcommand{\oppsit}{\oppsi\adj}

\newcommand{\oppsil}[1]{\paraml\oppsi{#1}}
\newcommand{\oppsilz}[1]{\paramlz\oppsi{#1}}

\newcommand{\oppsitl}[1]{\paraml{\oppsi\adj}{#1}}
\newcommand{\oppsitlz}[1]{\paramlz{\oppsi\adj}{#1}}

\newcommand{\oppsix}{\oppsi(x)}
\newcommand{\oppsixl}{\oppsil{x}}

\newcommand{\oppsinn}[1]{\oppsi_{#1}}
\newcommand{\oppsinnt}[1]{\oppsi_{#1}\adj}

\newcommand{\oppsin}{\oppsinn{n}}
\newcommand{\oppsina}{\oppsinn{\na}}
\newcommand{\oppsinr}{\oppsinn{\nr}}
\newcommand{\oppsinar}{\oppsinn{\nar}}

\newcommand{\oppsinnl}[2]{\paraml{\oppsinn{#1}}{#2}}
\newcommand{\oppsinnlz}[2]{\paramlz{\oppsinn{#1}}{#2}}

\newcommand{\oppsinntlz}[2]{\paramlz{\oppsinnt{#1}}{#2}}

\newcommand{\oppsinarl}[1]{\oppsinnl{\nar}{#1}}
\newcommand{\oppsinarlz}[1]{\oppsinnlz{\nar}{#1}}

\newcommand{\oppsinartlz}[1]{\oppsinntlz{\nar}{#1}}

\newcommand{\oppsinal}[1]{\oppsinnl{\na}{#1}}

\newcommand{\oppsinrl}[1]{\oppsinnl{\nr}{#1}}

\newcommand{\oppsinarxl}{\oppsinarl{x}}

\newcommand{\oppsinaxl}{\oppsinal{x}}

\newcommand{\oppsinrxl}{\oppsinrl{x}}

\newcommand{\oppsilll}[1]{\oppsi_{#1}}
\newcommand{\oppsinnlll}[2]{\oppsi_{#1,#2}}

\newcommand{\oppsill}{\oppsi_{\lambdaz}}
\newcommand{\oppsinnll}[1]{\oppsinnlll{#1}{\lambdaz}}
\newcommand{\oppsinll}{\oppsinnll{n}}
\newcommand{\oppsinall}{\oppsinnll{\na}}
\newcommand{\oppsinrll}{\oppsinnll{\nr}}
\newcommand{\oppsinarll}{\oppsinnll{\nar}}

\newcommand{\opPsi}{\op{\Psi}}

\begin{document}

\title{Axiomatic, Parameterized, Off-Shell QFT}


\author{Ed Seidewitz}
\email{seidewitz@mailaps.org}
\affiliation{14000 Gulliver's Trail, Bowie MD 20720 USA}

\date{21 November 2016}


\begin{abstract}
        
Axiomatic quantum field theory (QFT) provides a rigorous mathematical
foundation for QFT, and it is the basis for proving some important
general results, such as the well-known spin-statistics theorem.
Free-field QFT meets the axioms of axiomatic QFT, showing they are
consistent.  Nevertheless, even after more than 50 years, there is
still no known non-trivial theory of quantum fields with interactions
in four-dimensional Minkowski spacetime that meets the same axioms.

This paper provides a similar axiomatic basis for \emph{parameterized}
QFT, in which an invariant, fifth \emph{path parameter} is added to
the usual four spacetime position arguments of quantum fields.
Dynamic evolution is in terms of the path parameter rather than the
frame-dependent time coordinate.  Further, the states of the theory
are allowed to be \emph{off shell}.  Particles are therefore
fundamentally ``virtual'' during interaction but, in the appropriate
non-interacting, large-time limit, they dynamically tend towards
``physical'', on-shell states.

Unlike traditional QFT, it is possible to define a mathematically
consistent interaction picture in parameterized QFT. This may be used
to construct interacting fields that meet the same axioms as the
corresponding free fields.  One can then re-derive the Dyson series
for scattering amplitudes, but without the mathematical inconsistency
of traditional, perturbative QFT. The present work is limited to the
case of scalar fields, and it does not address remaining issues of
gauge symmetry and renormalization.  Nevertheless, it still
demonstrates that the parameterized formalism can provide a consistent
foundation for the interpretation of QFT as used in practice and,
perhaps, for better dealing with its further mathematical issues.

\end{abstract}

\maketitle
\titlepage


\section{Introduction} \label{sect:intro}

\subsection{Why Off Shell?} \label{sect:why-off-shell}

The momentum mass-shell condition $p^{2} = -m^{2}$, using a spacetime
metric signature of $(-+++)$, is the fundamental kinematic
relationship in special relativity.  It therefore seems obvious that
it should be carried over to relativistic quantum mechanics.  Indeed,
the operator form of this condition is just the Klein-Gordon equation
in quantum field theory (QFT), and it also determines the spectrum of
the relativistic momentum operator.

However, the state space of QFT still naturally includes both
\emph{on-shell} states that meet the mass-shell condition and
\emph{off-shell} states that do not.  The ``physical,'' on-shell
states need to be picked out as those satisfying the mass-shell
constraint.  Nevertheless, off-shell states cannot be completely
eliminated from the theory, since propagation is off shell along the
internal segments of Feynman diagrams generated using perturbative
QFT.

In the usual derivation (as given in any traditional QFT textbook,
such as \refcites{peskin95} or \cite{weinberg95}), off-shell particles
appear in the formalism because the end points of an internal segment
may have either time ordering.  However, physical particles are
considered to propagate forward in time, with the on-shell propagator
\begin{equation*} 
    \propa \equiv (2\pi)^{-3}\intthree p\,
	\frac{\me^{\mi[-\Ep(x^{0}-\xz^{0}) + 
	               \threep\cdot(\threex-\threexz)]}}
	     {2\Ep} \,,
\end{equation*}
for $x^{0} > \xz^{0}$ and $\Ep = \sqrt{\threep^{2} + m^{2}}$.  When
``backward'' time propagation is then re-interpreted as being for an
(on-shell) antiparticle propagating forward in time, and the particle
and antiparticle propagators are combined, the result is the
relativistically invariant Feynman propagator
\begin{equation}
    \begin{split} \label{eqn:A1}
	\prop &= \theta(x^{0}-\xz^{0})\propasym(x-\xz) 
		    + \theta(\xz^{0}-x^{0})\propasym(x-\xz)\conj \\
	      &= -\mi(2\pi)^{-4}\intfour p\,
		    \frac{\me^{\mi p\cdot(x - \xz)}}
			 {p^{2} + m^{2} - \mi\varepsilon} \,.
   \end{split}
\end{equation}
This allows for the transfer of any possible four-momentum, not just a
momentum that meets the mass-shell condition.

Such internal segments are said to represent \emph{virtual} particles
because, of course, ``real,'' physical particles must be on shell.
However, as Feynman noted \cite{feynman61}, ``real'' particles
actually always have a finite lifetime (at least if they are ever to
be detected), and so effectively propagate on finite-length internal
segments of Feynman diagrams.  As a result, they are all technically
virtual particles!

Of course, a particle does not have to propagate very far before the
interaction by which it is observed is essentially in the future light
cone of all classical observers of interest.  In this case, the
particle is unambiguously in the future of all such observers and can
be considered to be on shell.  But still, fundamentally, a particle is
only truly, exactly on shell in the limit of propagation into the
infinite future without interaction, at which point it can be
considered to be unambiguously in the future of \emph{all} potential
observers (except, of course, that without interaction it can never
really be observed).

Now, considerations such as the above may seem purely philosophical.
However, they do suggest the possibility that, perhaps, one could
formulate an interaction theory that starts by considering (so-called)
virtual, off-shell particles to be the ``true'' particles.  On-shell
particles would then be considered a limiting case of particles with
an infinite lifetime and a useful approximation for particles with
finite but ``long'' lifetimes.

The important benefit of this alternative viewpoint is that it
simplifies the treatment of the state space and allows for state
dynamics to be formulated in a way that is closely analogous to
non-relativistic quantum mechanics.  However, dynamic state evolution
could no longer be in terms of time, as it is in non-relativistic
quantum mechanics, because this would be relativistically frame
dependent.  Instead, evolution must be in terms of some additional,
invariant parameter.  But what, then, is this parameter, and how might
relativistic quantum mechanics be formulated in terms of it?

\subsection{Why Parameterized?} \label{sect:why-parameterized}

In non-relativistic quantum mechanics, the Hamiltonian operator $\opH$
determines the time evolution of a (Schr\"o\-dinger picture) wave
function $\psi(t)$ via the Schr\"odinger equation
\begin{equation*} 
    \mi\deriv{\psi(t)}{t} = \opH\psi(t)
\end{equation*}
(here and in the following I take $\hbar = 1$). Importantly, while 
$\opH$ also represents the observable energy of the system, there is 
no ``time operator'' -- time is not an observable in non-relativistic 
quantum mechanics, it is just an evolution parameter.

In relativistic quantum mechanics, the Hamiltonian can also be taken
to be the energy operator and the generator of time translations.
However, this operator is now just the time component $\opP^{0}$ of
the relativistic four-vector momentum operator $\opP$.  And, of
course, such a single component is not Lorentz invariant, so there can
be no relativistically invariant analog of the Schr\"odinger equation
using this definition of the Hamiltonian.  Further, the position
operator $\opX^{0}$ conjugate to $\opP^{0}$ \emph{is} a
(frame-dependent) observable for time, which, therefore, can no longer
be considered to be just an evolution parameter.

On the other hand, one can instead define a truly relativistic
Hamiltonian which, for a free particle, is $\opH = \opP^{2} + m^{2}$.
However, applying this operator to an on-shell particle state $\psi$
then gives $\opH\psi = 0$.  That is, the Hamiltonian vanishes
\emph{identically} when applied to on-shell states.  Particularly when
interactions and gravitation are included in $\opH$, this equation is
known as the \emph{Wheeler-DeWitt equation}, which is fundamental for
quantum gravity and cosmology \cite{dewitt67,hartle83}.

Thus, rather than being a generator of system evolution, the
Hamiltonian acts to impose a constraint on all physical states.  A
Hamiltonian-based mechanics for on-shell relativistic quantum
mechanics therefore requires the use of Dirac's theory of constraints
\cite{dirac50,dirac64}, a powerful but complicating generalization of
standard Hamiltonian mechanics.  And, in general, the Wheeler-DeWitt
equation is not easy to solve.  (For example, see the extensive
discussion in \refcite{rovelli04} on Hamiltonian constraint mechanics
and solving the Wheeler-DeWitt equation in the context of quantum
gravity.)

But now suppose we did \emph{not} require $\psi$ to be on shell.  In
this case we would have $\opH\psi \neq 0$.  Indeed, by analogy with
the non-relativistic case, we could then consider $\opH$ to generate
evolution in a relativistically invariant parameter $\lambda$, such that
\begin{equation} \label{eqn:A2}
    -\mi\deriv{\psi(\lambda)}{\lambda} = \opH\psi(\lambda) \,,
\end{equation}
a clear analog of the non-relativistic Schršdinger equation.  A
\emph{parameterized} formalism for relativistic quantum mechanics is
one that includes an invariant evolution parameter such as $\lambda$.

But what exactly is this evolution parameter?

To see, consider that, for wave functions in the position
representation, such a parameter acts as a fifth argument, in addition
to the usual four position arguments of Minkowski space.  Using the
relativistic Hamiltonian operator in the position representation,
\begin{equation*}
    \opH = -\frac{\partial^{2}}{\partial x^{2}} + m^{2} \,,
\end{equation*}
where $\partial^{2} / \partial x^{2}$ is the D'Alembertian operator
\begin{equation*}
    \frac{\partial^{2}}{\partial x^{2}} \equiv
	- \frac{\partial^{2}}{(\partial x^{0})^{2}}
	+ \frac{\partial^{2}}{(\partial x^{1})^{2}}
        + \frac{\partial^{2}}{(\partial x^{2})^{2}}
	+ \frac{\partial^{2}}{(\partial x^{3})^{2}}
\end{equation*}
(with $c = 1$), the solution to \eqn{eqn:A2} is
\begin{equation*}
    \psixl = \intfour \xz\, \kerneld \psixlz \,.
\end{equation*}
with the \emph{propagation kernel}
\begin{equation} \label{eqn:A3}
    \kerneld = (2\pi)^{-4} \intfour p \, \me^{\mi p\cdot(x - \xz)}
               \me^{-\mi (p^{2} + m^{2})(\lambda-\lambdaz)} \,.
\end{equation}
The wave function $\psixl$ is just the parameterized probability
amplitude function defined by Stueckelberg \cite{stueckelberg41}.  If
we take the \emph{path} of a particle to be a curve in spacetime,
$x^{\mu} = q^{\mu}(\lambda)$, $\mu = 0, 1, 2, 3$, with $\lambda$ as
its \emph{path parameter}, then the $\psixl$ represent the probability
amplitude for a particle to reach position $x$ at the point along its
path with parameter value $\lambda$.

This viewpoint of $\lambda$ as a path parameter is reinforced by the
observation, first made by Feynman \cite{feynman50}, that the
propagation kernel given in \eqn{eqn:A3} can be written in the form
of a spacetime path integral
\begin{equation} \label{eqn:A4}
    \kersym(x-\xz; \lambda_{1}-\lambdaz) =
        \eta \intDfour q\,
            \delta^{4}(q(\lambda_{1}) - x) \delta^{4}(q(\lambdaz) - \xz)
            \exp\left( 
                \mi \int^{\lambda_{1}}_{\lambdaz} \dl L(\qdotsq(\lambda))
            \right)
\end{equation}
for an appropriate normalization constant $\eta$ and the Lagrangian
function
\begin{equation*}
    L(\qdotsq) = \frac{1}{4}\qdotsq - m^{2} \,.
\end{equation*}
In the path integral above, the notation $\Dfour q$ indicates that the
integral is over the four functions $\qmul$ and the delta functions
constrain the starting and ending points of the paths integrated 
over. Further, if we integrate this over paths from $\xz$ to $x$ with 
all possible \emph{intrinsic lengths} $\lambda_{1} - \lambda_{0}$, 
the result is just the Feynman propagator \cite{feynman50,%
teitelboim82,halliwell01b,seidewitz06a},
\begin{equation} \label{eqn:A5}
    \prop = \int_{\lambdaz}^{\infty} \dif \lambda_{1}\, 
            \kersym(x-\xz; \lambda_{1}-\lambdaz) \,.
\end{equation}

Moreover, consider a state $\psi$ that maintains a fixed
three-momentum $\threep$, and integrate it over all paths that end at
some future time $t$.  Then, as shown in \refcite{seidewitz06a}, as $t
\rightarrow \infty$, the energy $p^{0} \rightarrow \sqrt{\threep^{2} +
m^{2}}$.  That is, in this limit, we dynamically recover the
mass-shell constraint $p^{2} = -(p^{0})^{2} + \threevec{p}^{2} =
-m^{2}$, as desired.

Thus, we can think of the $\lambda$ as parameterizing the position of
a particle along its path in spacetime.  Note, however, that such a
parameter is \emph{not} necessarily the proper time for the particle,
as has been sometimes suggested.  Indeed, constraining $\lambda$ to be
the proper time in the paths integrated over in \eqn{eqn:A4} would
mean constraining the particle four-velocity to be timelike, which is
equivalent to re-introducing the mass-shell constraint.  The result
would not then be the relativistic propagation kernel, and the
integral in \eqn{eqn:A5} would not produce the Feynman propagator.
Thus, for an off-shell theory, it is necessary to consider \emph{all}
possible paths, not just ones with timelike four-velocities.

There is actually a long history of approaches using such a fifth
parameter for relativistic quantum mechanics, going back to proposals
in the late thirties and early forties of Fock \cite{fock37} as well
as Stueckelberg \cite{stueckelberg41,stueckelberg42}.  The idea
appeared subsequently in the work of a number of well-known authors,
including Nambu \cite{nambu50}, Feynman \cite{feynman50, feynman51},
Schwinger \cite{schwinger51}, DeWitt-Morette \cite{morette51} and
Cooke \cite{cooke68}.  However, it was not until the seventies and
eighties that the theory was more fully developed, particularly by
Horwitz and Piron \cite{horwitz73, piron78} and Fanchi and Collins
\cite{collins78, fanchi78, fanchi83, fanchi93}, into what has come to
be called \emph{relativistic dynamics}.  The approach is particularly
applicable to the study of quantum gravity and cosmology, in which the
fundamental equations (such as the Wheeler-DeWitt equation) make no
explicit distinction for the time coordinate (see, e.g.,
\refcites{teitelboim82,hartle83,hartle86,hartle95,halliwell01a,%
halliwell01b,halliwell02}).

Extension of relativistic dynamics to a second-quantized QFT has been
somewhat more limited, focusing largely on application to quantum
electrodynamics \cite{fanchi79,saad89,land91,shnerb93,pavsic91,%
horwitz98,pavsic98}.  I have previously proposed a foundational
parameterized formalism for QFT and scattering based on spacetime
paths \cite{seidewitz06a,seidewitz06b,seidewitz09,seidewitz11}, along
the lines that I outlined above.  The purpose of the present paper is
to provide a more formal grounding for parameterized QFT, using an
axiomatic approach.  While the concept of paths will occasionally
still be helpful for intuitive motivation, the theory is presented
here entirely in field-theoretic mathematical language, without the
use of spacetime path integrals.

\subsection{Why Axiomatic?} \label{sect:why-axiomatic}

Beginning in the 1950s and 1960s, a number of researchers attempted to
provide rigorous foundations for QFT. G{\aa}rding and Wightman
introduced a set of precise axioms for QFT on Minkowski space
\cite{wightman59,wightman64}, and Haag and Kastler developed a related
approach to local functions of the field \cite{haag93}.  These
approaches had some early important successes, including modeling
free-field theory, proving that a field theory can be reconstructed
from its vacuum expectation values and rigorously establishing the
link between spin and statistics (see, for example,
\refcites{streater64} and \cite{bogolubov75}).  This early work also
has led to considerable subsequent mathematical work on algebraic,
Euclidean and constructive QFT.

In addition, axiomatic QFT provided the basis for the first rigorous
proof of Haag's theorem.  This theorem states that, under the Wightman
axioms for QFT, any field that is unitarily equivalent to a free field
must itself be a free field \cite{haag55, hall57, streater64}.  This
is troublesome, because the usual Dyson perturbation expansion of the
scattering matrix is based on the interaction picture, in which the
interacting field is presumed to be related to the free field by a
unitary transformation.  And, as Streater and Wightman note
\cite{streater64}, Haag's theorem means that such a picture should not
exist in the presence of actual interaction.

Indeed, after more than 50 years, there is still no known,
non-trivial, interacting QFT in four-dimensional Minkowski space that
satisfies all of the Wightman axioms.  While the existence of such a
theory has not been ruled out (and there has been progress for
simplified cases, such as in fewer dimensions), this situation does
suggest that perhaps the Wightman axioms do not provide the most
useful formal grounding for interacting QFT after all.

In this paper, I propose a set of axioms for parameterized QFT that
are inspired by the Wightman axioms, but differ in some important
ways.  In particular, the underlying states of the theory are not
required to be on shell, which means that the spacetime symmetry group
of a particle field no is no longer restricted to just a single
irreducible representation of the Poincar\'e group.  This also
requires a different treatment of Hamiltonian operators, which act as
generators of parameter evolution (as foreshadowed in
\sect{sect:why-parameterized}), and for which, also, a vacuum state is
only unique relative to a given Hamiltonian.

\Sect{sect:theory} presents the formal axiomatic theory of
parameterized QFT, starting with the axioms for the Hilbert space of
states and for the field operators defined on those states.  Not
surprisingly, it is straightforward to construct a theory of free
fields that meets the given axioms.  However, it turns out that, in 
the
parameterized theory, it is, in fact, possible to also
straightforwardly construct interacting fields that meet the axioms.
The critical result is that the parameterized theory has no analog of
Haag's theorem (for a more in-depth look at why this is so, see also
\refcite{seidewitz15}).  This means that parameterized QFT has a
mathematically consistent interaction picture, allowing for the
construction of an interacting field using a unitary transformation
from a free field, so that it naturally meets the field axioms.

\Sect{sect:application} then takes up the task of applying the theory
to the modeling of interactions and computing scattering amplitudes.
This development is not complete, because it includes only cursory
consideration, at this point, of regularization and renormalization,
which are still required to obtain finite results for scattering
amplitudes.  Nevertheless, it is shown that the parameterized theory
can formally reproduce, term for term, the usual Dyson expansion for a
scattering amplitude obtained using traditional perturbative QFT. And,
since the use of the interaction picture is mathematically consistent
in parameterized QFT, this actually explains how the perturbative
expansions in traditional QFT can produce such empirically accurate
results, even in the face of Haag's theorem.

For simplicity, only massive, scalar fields will be considered in this
paper.  However, the parameterized formalism has been extended to
non-scalar fields in other work (see, for example,
\cite{seidewitz09}), and there is no reason to believe this could not
be incorporated into the axiomatic formalism, much as for traditional
QFT. On the other hand, including massless and gauge fields in the
parameterized axiomatic theory is more difficult and is an important
avenue for future work (see also \refcites{saad89} and \cite{shnerb93}
for some earlier work addressing gauge theory in the context of
parameterized quantum electrodynamics).

\section{Theory} \label{sect:theory}

This section describes the basic theory of parameterized QFT.
\Sects{sect:hilbert} and \ref{sect:fields} present the axioms of
parameterized QFT, generally paralleling the axioms for traditional
QFT as presented in \cite{streater64}.  \Sect{sect:free} then
demonstrates the consistency of the axioms by constructing a theory of
free fields that satisfies them, and \sect{sect:expectation} discusses
general results related to the vacuum expectation values of fields.
\Sect{sect:interacting} shows how the theory can also cover the case
of interacting fields, a topic which will be addressed further in
\sect{sect:application}.  \Sect{sect:integrated} develops the formalism
for integrating fields over the path parameter and
\sect{sect:antiparticles} addresses the related issue of the modeling
of antiparticles.

\subsection{Hilbert Space} \label{sect:hilbert}

As in traditional QFT, the states in the theory, over which fields
will be defined, are given by the unit rays in a separable Hilbert
space $\HilbH$.  For any vectors $\Psi, \Phi \in \HilbH$, the inner
product defined on $\HilbH$ is denoted $(\Psi, \Phi)$.  For all $\Psi
\in \HilbH$, $\Psi^{2} \equiv (\Psi, \Psi)$ must be finite and
non-negative (i.e., vectors are normalizable), and $\Psi^{2} = 0$ if
and only if $\Psi = 0$.

We can introduce the usual Dirac bra and ket notation by first
extending $\HilbH$ to a rigged Hilbert space, or Gel'fand triple
\cite{gelfand55,gelfand64}, $\HilbK \subset \HilbH \subset \HilbK'$.
Here, $\HilbK$ is a nuclear space that is a dense subset of $\HilbH$
\cite{bogolubov75}.  This is the space of kets, $\ket{\Phi} \in
\HilbK$.  $\HilbK'$ is the space of continuous, linear, complex-valued
functionals $F[\Phi]$ on $\HilbK$, known as \emph{distributions}.
This is the space of bras, $\bra{F} \in \HilbK'$.  (For a complete
treatment of bras, kets and Gel'fand triples, see \refcite{pelc97}.)

For any $\bra{F} \in \HilbK'$, the action of $\bra{F}$ on a
ket $\ket{\Phi}$ is denoted by the bracket
\begin{equation*}
    \inner{F}{\Phi} \equiv F[\Phi] \,.
\end{equation*}
For any $\Psi \in \HilbH$, there is a dual bra $\bra{\Psi}$ whose
action is defined by
\begin{equation*}
    \inner{\Psi}{\Phi} = (\Psi, \Phi) \,.
\end{equation*}
It is under this duality mapping that $\HilbH$ may be considered a
subset of $\HilbK'$.

The following axiom establishes the basic assumptions about $\HilbH$,
including its relativistic nature.  The assumed relativistic
transformation law is the same as in traditional relativistic quantum
mechanics, but the axiom does \emph{not} include the usual spectral
conditions on the spacetime translation operator.

\setcounter{axiom}{-1}
\begin{axiom}[Relativistic States] \label{axm:rqm}
    The states of the theory are the unit rays in a separable Hilbert
    space $\HilbH$ on which is defined a unitary representation
    $\{\Delta x, \Lambda\} \rightarrow \opU(\Delta x, \Lambda)$ of the
    Poincar\'e group.
\end{axiom}

Since the spacetime translation operator $\opU(\Delta x, 1)$ is
unitary, it can be written as
\begin{equation*}
    \opU(\Delta x, 1) = \me^{\mi\opP \cdot \Delta x} \,,
\end{equation*}
where $\opP$ is a Hermitian operator (i.e., $\opP = \opP\adj$).  As
usual, we will interpret $\opP$ as the relativistic energy-momentum
operator.  However, as noted above, we will \emph{not} require that
the eigenvalues of $\opP$ be on-shell.  \Sect{sect:application} will
address the emergence of physical on-shell states in the context of
scattering processes.

We will also \emph{not} take $\opP^{0}$, the (frame-dependent) energy
operator and generator of time translations, to be the Hamiltonian
operator.  Instead, we allow the separate identification of one or 
more Hamiltonian operators on $\HilbH$, as long as they meet the 
following definition.

\begin{definition*}[Hamiltonian operator] 
    An operator $\opH$ on $\HilbH$ with the following properties:
    \begin{enumerate}
	\item $\opH$ is Hermitian.
	
	\item $\opH$ commutes with all spacetime transformations
	$\opU(\Delta x, \Lambda)$ (and, hence, with $\opP$).
	
	\item $\opH$ has an eigenstate $\ketvac \in \HilbK \subset
	\HilbH$ such that $\opH\ketvac = 0$, and this is the unique
	null eigenstate for $\opH$, up to a constant phase.
    \end{enumerate}
\end{definition*}

Note that this definition essentially introduces the concept of a
\emph{vacuum state}, but only relative to a choice of Hamiltonian
operator.  Each Hamiltonian operator must have a unique vacuum state,
but different Hamiltonian operators defined on the same Hilbert space
may have different vacuum states.  On the other hand, the Hamiltonian
operator is \emph{not} required to be positive definite, since it is
no longer considered to represent the energy observable, so its vacuum
state is not a ``ground state'' in the traditional sense.  However,
when $\HilbH$ is a Fock space with a particle interpretation, the
vacuum state for an identified Hamiltonian can be considered to be the
``no particle'' state (as will be the case for the free-field theory
defined in \sect{sect:free}).

Rather than generating time translation, a Hamiltonian operator $\opH$
is taken instead to generate evolution in the parameter $\lambda$, as
given by the exponential operator $\me^{-\mi\opH\lambda}$.  Since
$\opH$ commutes with $\opP$, momentum is conserved under evolution in
$\lambda$.  The parameter $\lambda$ itself is \emph{not} considered to
be an observable in the theory, but is, rather, treated as just an
evolution parameter, similarly to time in non-relativistic quantum
mechanics.

By further analogy with non-relativistic theory, we can define both
Schr\"odinger and Heisenberg pictures for evolution in $\lambda$.
That is, a Schr\"odinger-picture state $\Psi_{S}(\lambda)$ evolves in
$\lambda$ and is related to the corresponding Heisenberg-picture state
$\Psi_{H}$ by $\Psi_{S}(\lambda) = \me^{-\mi\opH\lambda} \Psi_{H}$,
while a Heisenberg-picture operator $\opA_{H}(\lambda)$ evolves in
$\lambda$ and is related to the corresponding Schr\"odinger-picture
operator $\opA_{S}$ by $\opA_{H}(\lambda) = \me^{\mi\opH\lambda}
\opA_{S} \me^{-\mi\opH\lambda}$, so that $\bra{\Psi_{H}}
\opA_{H}(\lambda) \ket{\Phi_{H}} = \bra{\Psi_{S}(\lambda)} \opA_{S}
\ket{\Phi_{S}(\lambda)}$, for all $\Psi$, $\Phi$ and $\lambda$.

The evolution of Schr\"odinger-picture states can also be given by the
\emph{Stueckelberg-Schr\"odinger equation}
\begin{equation*}
    \mi\deriv{\Psi_{S}}{\lambda} = \opH\Psi_{S} \,.
\end{equation*}
Similarly, Heisenberg-picture operators evolve according to
\begin{equation*}
    \mi\deriv{\opA_{H}}{\lambda} = [\opA, \opH] \,.
\end{equation*}

The single null eigenstate of a Hamiltonian $\opH$ is a \emph{vacuum
state} for $\HilbH$, invariant in the $\lambda$ evolution generated by
$\opH$.  Such a state must also be invariant under Poincar\'e
transformations, since $\opH\opU(\Delta x, \Lambda)\ketvac =
\opU(\Delta x, \Lambda)\opH\ketvac = 0$, so $\opU(\Delta x,
\Lambda)\ketvac$ must equal $\ketvac$ (up to a phase, which repeated
application shows must be $1$).  Note, however, that $\opH$ is not
prohibited from having other eigenstates that are invariant under
Poincar\'e transformations, but evolve in $\lambda$.

\subsection{Fields} \label{sect:fields}

As in traditional QFT, we can define a \emph{field} as an operator
$\oppsix$ on $\HilbH$, which is a function of spacetime position $x$
(the present paper only covers the case of scalar fields).  For the
parameterized theory, though, we can also consider both
Schr\"odinger-picture and Heisenberg-picture versions of the field
operators, relative to parameter evolution according to a given
Hamiltonian $\opH$.  Taking $\oppsix$ to be a Schr\"odinger-picture
operator, the corresponding Heisenberg-picture operator is
\begin{equation*}
    \oppsixl 
	\equiv \me^{\mi\opH\lambda} \oppsix \me^{-\mi\opH\lambda} \,.
\end{equation*}
That is, in the (parameterized) Heisenberg picture, the field
operators $\oppsixl$ are functions of both the spacetime position
\emph{and} the parameter $\lambda$.

However, fields that are functions of position are generally too
singular to be well-defined operators on $\HilbH$
\cite{streater64,bogolubov75}.  Instead, to obtain mathematically
well-defined entities, one must \emph{smear} the fields with a
\emph{test function} $f(x)$ from the space $\HilbD$ of smooth,
parameterized, complex-valued functions with compact support in
spacetime and a norm
\begin{equation} \label{eqn:B1}
    |f|^{2} \equiv \intfour x\, f\conj(x)f(x)
\end{equation}
that is finite.

So, for any $f \in \HilbD$, formally define the operator-valued
functional
\begin{equation*}
    \oppsi[f] = \intfour x\, f\conj(x) \oppsix
              = \intfour p\, f\conj(p) \oppsi(p)
\end{equation*}
and its adjoint
\begin{equation*}
    \oppsi\adj[f] = \intfour x\, f(x) \oppsit(x)
                  = \intfour p\, f(p) \oppsit(p) \,.
\end{equation*}
The position and momentum representations of the functions $f$ and
fields $\oppsi$ are related by four-dimensional Fourier transforms:
\begin{equation*}
    f(p) = (2\pi)^{-2} \intfour x\, \me^{-\mi p \cdot x} f(x)
\end{equation*}
and
\begin{equation*}
    \oppsi(p) = 
	(2\pi)^{-2} \intfour x\, \me^{-\mi p \cdot x} \oppsix \,.
\end{equation*}
It is the smeared fields $\oppsi[f]$ and $\oppsi\adj[f]$ that are
actually well-defined operators in the theory. 

The Heisenberg-picture versions of the smeared fields are, similarly,
\begin{equation*}
    \param{\oppsi}[f\lambda] 
	= \intfour x\, f\conj(x) \oppsixl
	= \intfour p\, f\conj(p) \oppsil{p}
\end{equation*}
and its adjoint
\begin{equation*}
    \param{\oppsi\adj}[f\lambda] 
	= \intfour x\, f(x) \oppsitl{x}
        = \intfour p\, f(p) \oppsitl{p} \,,
\end{equation*}
for any value of the path parameter $\lambda$.  Note that the
Heisenberg-picture fields are smeared over the spacetime argument, but
not the parameter argument.

\renewcommand{\theaxiom}{\Roman{axiom}}
\begin{axiom}[Domain and Continuity of Fields] \label{axm:domain}
    For all $f \in \HilbD$, the operators $\oppsi[f]$ and their
    adjoints $\oppsit[f]$ are defined on a domain $D$ of states dense
    in $\HilbH$.  The $\opU(\Delta x, \Lambda)$, any Hamiltonian
    $\opH$, $\oppsi[f]$ and $\oppsit[f]$ all carry vectors in $D$
    into vectors in $D$.
\end{axiom}

\noindent 
The domain $D$ always contains the domain $D_{0}$ of states obtained
by applying polynomials in the fields and their adjoints to the vacuum
state of a given Hamiltonian $\opH$. Typically, we will be able to 
assume that $D_{0} = D$.

\begin{axiom}[Field Transformation Law] \label{axm:transform}
    For any Poincar\'e transformation $\{\Delta x, \Lambda\}$, for 
    any $\Psi \in D$,
    \begin{equation*}
	\opU(\Delta x, \Lambda) \oppsi[f] 
	\opU^{-1}(\Delta x, \Lambda) \Psi
	    = \oppsi[\{\Delta x, \Lambda\}f]\Psi \,,
    \end{equation*}
    where
    \begin{equation*}
	(\{\Delta x, \Lambda\}f)(x) \equiv
	    f(\Lambda^{-1}(x - \Delta x)) \,.
    \end{equation*}
\end{axiom}

\noindent 
This field transformation law is consistent with the spacetime
transformation law for states given in \axm{axm:rqm} (note again that
only scalar particles are being considered in the present paper).
While the axiom is given in terms of the Schr\"odinger-picture field,
the commutivity of $\opU(\Delta x, \Lambda)$ and $\opH$ implies that
it also applies to the Heisenberg-picture fields.

For the unsmeared field $\oppsix$, the above transformation law
for a Poincar\'e transformation $\{\Delta x, \Lambda\}$ gives
\begin{equation*}
    \opU(\Delta x, \Lambda) \oppsix \opU^{-1}(\Delta x, \Lambda)
	= \oppsi(\Lambda x + \Delta x) \,.
\end{equation*}
Taking the limit of infinitesimal $\Delta x$ with $\Lambda = 1$ then
gives
\begin{equation*}
    \mi \pderiv{}{x}\oppsi(x) = [\oppsi(x), \opP] \,.
\end{equation*}

\begin{axiom}[Commutation Relations] \label{axm:commutation}
    The field $\oppsi$ and its adjoint satisfy the (bosonic)
    commutation relations, for all $f', f \in \HilbD$ and any $\Psi
    \in D$,
    \begin{equation*}
	[\,\oppsi[f'], \oppsi[f]\,] \Psi 
	    = [\,\oppsi\adj[f'], \oppsit[f]\,] \Psi = 0
    \end{equation*}
    and
    \begin{equation*}
	[\,\oppsi[f'], \oppsit[f]\,] \Psi = (f',f) \Psi \,,
    \end{equation*}
    where
    \begin{equation*}
	(f', f) \equiv \intfour x\, {f'}\conj(x)f(x) \,.
    \end{equation*}
\end{axiom}

\noindent 
Note that $(f', f) = (|f'+f|^{2} - |f'|^{2} - |f|^{2})/2$, and so this
is well-defined if the norm given in \eqn{eqn:B1} is.

In terms of the unsmeared field, \axm{axm:commutation} gives the
four-dimensional commutation relation:
\begin{equation} \label{eqn:B2}
    [\oppsi(x'), \oppsit(x)] = \delta^{4}(x'-x) \,,
\end{equation}
or, for the Heisenberg-picture fields, the \emph{equal-$\lambda$}
commutation relation:
\begin{equation*} 
    [\oppsil{x'}, \oppsitl{x}] = \delta^{4}(x'-x) \,.
\end{equation*}
As a four-dimensional commutation relation, this is stronger than the
usual ``local commutivity'' axiom for traditional fields, which only
requires that fields commute when the positions $x'$ and $x$ are
spacelike.  In contrast, \eqn{eqn:B2} requires that fields commute
for all $x'$ and $x$ other than $x' = x$.

Further, since the Heisenberg-picture field operators are smeared over
the four-position $x$, but not the parameter $\lambda$, there is no
mathematical issue with going from the unsmeared form of the
commutation relations to the more rigorously defined smeared form.
This is in contrast to the equal-time commutation relations commonly
imposed in traditional QFT, which would require fields to make sense
as operators when smeared over only the three-position, an additional
strong, and possibly questionable, assumption (see the discussion in
\refcite{streater64}, p.  101).

\begin{axiom}[Cyclicity of the Vacuum] \label{axm:cyclicity}
    If $\opH$ is a Hamiltonian operator, then its vacuum state
    $\ketvac$ is in the domain $D$ of the field operators, and
    polynomials in the fields $\oppsi[f]$ and their adjoints
    $\oppsit[f]$, when applied to $\ketvac$, yield a set $D_{0}$ of
    states dense in $\HilbH$.
\end{axiom}

An operator $\opA$ constructed as a polynomial in the smeared fields
is a functional 
\begin{equation} \label{eqn:B3}
    \opA[\{f_{i}\}, \{g_{j}\}] = 
	\opA(\{\oppsi[f_{i}]\}, \{\oppsit[g_{j}]\}) \,,
\end{equation}
where the functions $f_{i}$ are arguments of the field operators and
the functions $g_{j}$ are arguments of their adjoints.  Such
polynomial operators clearly form an algebra $\HilbA$ under addition,
multiplication and scalar multiplication.  Further, $\HilbA$ is a
*-algebra, with involution given by the adjoint operation
\begin{equation} \label{eqn:B4}
    \opA\adj[\{g_{j}\}, \{f_{i}\}] \equiv
	(\opA[\{f_{i}\}, \{g_{j}\}])\adj \,.
\end{equation}

\subsection{Free Fields} \label{sect:free}

This section presents a theory of free fields that satisfy the axioms
presented in \sect{sect:fields}.  \emph{Free} fields represent
particles that do not interact.  Therefore, as in traditional QFT, we
can develop a theory of parameterized free fields that act on a Fock
space constructed as the direct sum
\begin{equation*}
    \HilbH = \bigoplus_{N=0}^{\infty}\HilbH^{(N)} \,,
\end{equation*}
where $\HilbH^{(N)}$ is the subspace of those states with exactly $N$ 
particles.

Let $\HilbH^{(N)}$ be the Hilbert space of functions $\Psi(\argNx)$,
symmetric in the interchange of any two arguments (for the case of
bosonic scalar particles considered in this paper), taken from the
space $\HilbL^{2}(R^{4N}, C)$ of complex-valued, square-integrable
functions on $R^{4N}$.  The inner product is given by
\begin{equation*} 
    (\Psi_{1}, \Psi_{2}) \equiv
	\intfour x_{1} \cdots \intfour x_{N}
	    \Psi_{1}(\argNx)\conj\Psi_{2}(\argNx) \,.
\end{equation*}
$N$-particle states are represented by $\Psi \in \HilbH^{(N)}$ that
are normalized such that
\begin{equation*}
    (\Psi, \Psi) = 
	\intfour x_{1} \cdots \intfour x_{N}
	    \Psi(\argNx)\conj\Psi(\argNx) = 1 \,.
\end{equation*}
The spacetime transformation law for these states is
\begin{equation*}
    \opU(\Delta x, \Lambda)\Psi(\argNx) 
	= \{\Delta x, \Lambda\}\Psi(\argNx)
	= \Psi(\Lambda^{-1}(x_{1} - \Delta x), \dots, 
	       \Lambda^{-1}(x_{N} - \Delta x)) \,.
\end{equation*}
The Fock space $\HilbH$ of states of any number of particles then, by
construction, satisfies the assumptions of \axm{axm:rqm}.

Next, consider $\HilbH^{(N)}$ as being in the Gel'fand triple
$\HilbK^{(N)} \subset \HilbH^{(N)} \subset \HilbK^{(N)\prime}$.  Here,
$\HilbK^{(N)}$ is a dense subset of $\HilbH^{(N)}$ consisting of
smooth functions $f(\argNx)$ with compact support, and
$\HilbK^{(N)\prime}$ is the space of continuous, linear,
complex-valued functionals $F[f]$ on $\HilbK^{(N)}$.  $\HilbH$ is then
in the Gel'fand triple $\HilbK \subset \HilbH \subset \HilbK'$, where
$\HilbK = \bigoplus_{N = 0}^{\infty} \HilbK^{(N)}$ and $\HilbK'=
\bigoplus_{N = 0}^{\infty} \HilbK^{(N)\prime}$.

The subspace $\HilbH^{(0)} \subset \HilbH$ contains the single
zero-particle state $\ketvac$, and we require that any \emph{free
Hamiltonian} have this state as its vacuum state.  As for traditional
QFT, we can then define a parameterized field theory of free particles
that satisfies the axioms given in \sect{sect:theory}, by having the
fields build up the Fock space $\HilbH$ from the vacuum state
$\ketvac$.  The fields $\oppsi[f]$ act as particle \emph{destruction}
operators, while their adjoints $\oppsit[f]$ act as particle
\emph{creation} operators.

Specifically, let $\ket{\Psi} \in \HilbK^{(N)}$, for $N > 0$.  Then
free fields are defined such that
\begin{equation*}
    \begin{aligned}
	\oppsi[f]\,\ket{0} &= 0 \,, \\
	\oppsi[f]\,\ket{\Psi}
	    &= \sum_{i=1}^{N}\op{f_{i}}\ket{\Psi}
		\in \HilbK^{(N-1)} \,,
    \end{aligned}
\end{equation*}
where the operators $\op{f}_{i}$ act as
\begin{equation*}
    (\op{f_{i}}\Psi)(\argn{x}{i-1},\argmn{x}{i+1}{N}) =
	\intfour x_{i}\, f\conj(x_{i})
	    \Psi(\argn{x}{i},\ldots,x_{N}) \,,
\end{equation*}
and
\begin{equation*}
    \begin{aligned}
	\oppsit[f]\,\ket{0} &= \ket{f} \in \HilbK^{(1)} \,, \\
	\oppsit[f]\,\ket{\Psi}
	    &= \mathrm{Sym} (\ket{f}\otimes\ket{\Psi})
		\in \HilbK^{(N+1)} \,,
    \end{aligned}
\end{equation*}
where $\mathrm{Sym}(\ldots)$ indicates symmetrization on function 
arguments.

In this way, we can build up any element of $\HilbK^{(N)}$, for any $N
> 0$, and therefore $\HilbK$, by applying a polynomial in
$\oppsit[f]$.  We can thus take $D = D_{0} = \HilbK$ and satisfy
\axms{axm:domain} and \ref{axm:cyclicity}.  Further, inserting fields
transformed under $\opU(\Delta x, \Lambda)$, it is clear that these
fields also satisfy \axm{axm:transform}.  Finally,
\begin{equation*}
    \begin{split}
	\oppsi[f']\oppsit[f]\,\ket{\Psi}
	    &= \sum_{i=1}^{N}\op{f'_{i}}\,
		\mathrm{Sym} (\ket{f}\otimes\ket{\Psi}) \\
	    &= (f',f)\ket{\Psi} +
	        \mathrm{Sym}(\ket{f} \otimes 
		    \sum_{i=1}^{N-1}
			\op{f'_{i}}\ket{\Psi}) \\
	    &= (f',f)\ket{\Psi} +
		\oppsit[f]\oppsi[f']\,\ket{\Psi} \,,
    \end{split}
\end{equation*}
so 
\begin{equation*}
    [\,\oppsi[f'], \oppsit[f]\,] \ket{\psi} 
	= (f',f) \ket{\psi} \,,
\end{equation*}
satisfying \axm{axm:commutation}.

\subsection{Vacuum Expectation Values} \label{sect:expectation}

One of the central results of traditional axiomatic QFT is that a
field theory can be reconstructed from its vacuum expectation values,
and this result can largely be carried over to parametrized QFT. To do
so, let $\psi[f]$ be a field defined on a Hilbert space $\HilbH$ and
$\opH$ be a Hamiltonian with vacuum state $\ketvac \in \HilbH$.  Then
the vacuum expectation values for the field are given by the 
\emph{Wightman distributions}
\begin{equation*}
    \begin{split}
	\fnW_{(m,n)}(\argn{f}{m}; \argn{g}{n})
	    &\equiv \bravac \oppsi[f_{1}]\cdots\oppsi[f_{m}]
		    \oppsit[g_{1}]\cdots\oppsit[g_{n}]\ketvac \\
	    &=      \bravac \param{\oppsi}[{f_{1}}\lambda]\cdots
	                    \param{\oppsi}[{f_{m}}\lambda]
		            \param{\oppsit}[{g_{1}}\lambda]\cdots
	                    \param{\oppsit}[{g_{m}}\lambda]\ketvac \,,
    \end{split}
\end{equation*}
for any value of $\lambda$.  Note that the definitions here have been
adapted for non-Hermitian, scalar field operators.

These distributions have the following properties.

\begin{enumerate}
    \item \emph{Relativistic Invariance.} For any Poincar\'e 
    transformation $\{\Delta x, \Lambda\}$,
    \begin{equation*}
	\fnW_{(m,n)}[\argn{f}{m}; \argn{g}{n}] \\
	    = \fnW_{(m,n)}[\{\Delta x, \Lambda\}f_{1}, \dots,
	                   \{\Delta x, \Lambda\}f_{m};
	                   \{\Delta x, \Lambda\}g_{1}, \dots,
	                   \{\Delta x, \Lambda\}g_{n}] \,.
    \end{equation*}
    \item \emph{Hermiticity.}
    \begin{equation*}
	\fnW_{(m,n)}[\argn{f}{m}; \argn{g}{n}]
	    = \fnW\conj_{(n,m)}[\argmn{g}{n}{1}; \argmn{f}{m}{1}] \,.
    \end{equation*}
    \item \emph{Commutivity.} For any $j$ or $k$,
    \begin{equation*}
	\begin{split}
	    \fnW_{(m,n)}[\argn{f}{m}; \argn{g}{n}]
		&= \fnW_{(m,n)}[\argn{f}{j+1}, \argmn{f}{j}{n};
		                \argn{g}{n}) \\
		&= \fnW_{(m,n)}[\argn{f}{m}; 
				\argn{g}{k+1}, \argmn{g}{k}{m}]
	\end{split}
    \end{equation*}
    and, if $(f_{m}, g_{1}) = 0$,
    \begin{equation*}
	\fnW_{(m,n)}[\argn{f}{m}; \argn{g}{n}]
	    = \fnW_{(m,n)}(\argn{f}{m-1}, g_{1}; 
			   f_{m}, \argmn{g}{2}{n}] \,.
    \end{equation*}
    \item \emph{Positive Definiteness.} For any sequence of test
    functions $\{f_{11}, f_{21}, f_{22}, f_{31}, \dots\}$ and each
    integer $n \geq 0$,
    \begin{equation*}
	\sum_{j = 0, k = 0}^{n} 
	    \fnW_{(j,k)}[f\conj_{jj}, \dots, f\conj_{j1}; 
	                 f_{k1}, \ldots, f_{kk}]
	    \geq 0 \,.
    \end{equation*}
\end{enumerate}

These are similar properties to those for the Wightman distributions
in traditional QFT \cite{streater64,bogolubov75}, except that local
commutivity has been replaced with a stronger commutivity condition
and there are no spectral conditions.  Therefore, an analogous
Reconstruction Theorem still holds, in the following form.

\begin{theorem}[Reconstruction Theorem] \label{thm:reconstruction}
    Let $\{\fnW_{(m,n)}[\argn{f}{m}; \argn{g}{n}]\}$, for $m,n = 1, 2,
    \dots$, be a sequence of distributions, depending on the $m+n$
    test functions $\argn{f}{m}$ and $\argn{g}{n}$, satisfying
    properties 1 to 4 for Wightman distributions.  Then there exists a
    separable Hilbert space $\HilbH$, a continuous unitary
    representation $\opU(\Delta x, \Lambda)$ of the Poincar\'e group
    on $\HilbH$, a state $\ketvac \in \HilbH$ that is invariant under
    $\opU(\Delta x, \Lambda)$ and a field $\oppsi[f]$ on $\HilbH$ with
    domain $D$, such that
    \begin{equation*}
	\bravac \oppsi[f_{1}]\cdots\oppsi[f_{m}]
	        \oppsit[g_{1}]\cdots\oppsit[g_{n}]\ketvac
		= \fnW_{(m,n)}[\argn{f}{m}; \argn{g}{n}] \,.
    \end{equation*}
    Further, if $\HilbH'$ is a Hilbert space, $\opU'(\Delta x,
    \Lambda)$ is a continuous representation of the Poincar\'e group
    on $\HilbH'$, $\ket{0'} \in \HilbH'$ is invariant under
    $\opU'(\Delta x, \Lambda)$ and $\oppsi'[x]$ is a field on 
    $\HilbH'$ with domain $D'$ and the property
    \begin{equation*}
	\bra{0'} \oppsi'[f_{1}]\cdots\oppsi'[f_{m}]
	         \oppsi'{}\adj[g_{1}]\cdots\oppsi'{}\adj[g_{n}]\ket{0'}
		= \fnW_{(m,n)}[\argn{f}{m}; \argn{g}{n}] \,,
    \end{equation*}
    then there exists a unitary transformation $\opG$ from $\HilbH$ 
    onto $\HilbH'$, such that
    \begin{equation*}
	\begin{split}
	    \opU'(\Delta x, \Lambda) 
		&= \opG\opU(\Delta x, \Lambda)\opG^{-1} \,, \\
	    \oppsi'[f] &= \opG\oppsi[f]\opG^{-1} \,, \\
	    \ket{0'} &= \opG\ketvac \textrm{ and} \\
	    D' &= \opG D \,.
	\end{split}
    \end{equation*}
\end{theorem}

The proof of this theorem is essentially the same as for traditional
QFT \cite{streater64,bogolubov75}.  However, the spectral conditions
required on traditional Wightman functions are necessary for proving
the cluster decomposition property for those functions, which is in
turn used to prove the uniqueness of the constructed vacuum state
\cite{streater64}.  Since the Wightman functions for the parameterized
theory do not have spectral restrictions, the Reconstruction Theorem
as stated above does not claim that $\ketvac$ is necessarily unique.

However, suppose that $\HilbH$ and $\HilbH'$ are Hilbert spaces
related by a unitary transformation $\opG$, with corresponding vacuum
states $\ketvac$ and $\ket{0'}$ and fields $\oppsi$ and $\oppsi'$, as
in the second part of the Reconstruction Theorem.  Suppose further
that $\opH$ is a Hamiltonian defined on $\HilbH$ with $\ketvac$ as its
unique vacuum state.  The corresponding Hamiltonian operator in
$\HilbH'$ is $\opH' = \opG\opH\opG^{-1}$.  Then $\opH'\ket{0'} =
\opH'\opG\ketvac = \opG\opH\ketvac = 0$, and, since $\opH'$ has the
same spectrum as $\opH$, $\ket{0'}$ must be the unique null eigenstate
of $\opH'$.

In addition, if $\param{\oppsi}[f\lambda]$ is the Heisenberg-picture
version of $\oppsi[f]$, using the Hamiltonian $\opH$, then the
Heisenberg-picture for $\oppsi'[f]$, using $\opH'$, is
\begin{equation*}
    \begin{split}
	\param{\oppsi'}[f\lambda]
	    &= \me^{\mi\opH'\lambda}\oppsi'[f]
	       \me^{-\mi\opH'\lambda} \\
	    &= \me^{\mi\opH'\lambda}
	       \opG\me^{-\mi\opH\lambda}
	       \param{\oppsi}[f\lambda]
	       \me^{\mi\opH\lambda}\opG^{-1}
	       \me^{-\mi\opH'\lambda} \\ 
	    &= \opG\param{\oppsi}[f\lambda]\opG^{-1} \,,
    \end{split}
\end{equation*}
where the last equality follows from $\opH'\opG = \opG\opH$. So the 
Heisenberg-picture forms of the fields are also unitarily related by 
$\opG$.

\subsection{Interacting Fields} \label{sect:interacting}

Let the field $\oppsi$ and Hamiltonian $\opH$ be unitarily related to
the field $\oppsi'$ and Hamiltonian $\opH'$, as discussed in
\sect{sect:expectation}.  Now, however, suppose that the fields are
both defined on the \emph{same} Hilbert space $\HilbH$.  Then the
unitary transformation $\opG$ maps $\HilbH$ onto $\HilbH$ and, for all
$\Psi \in \HilbH$, it is also the case that $\Psi' = \opG\Psi \in
\HilbH$. In the Schr\"odinger picture, $\Psi_{S}(\lambda)$ evolves 
according to $\opH$, while $\Psi'_{S}(\lambda)$ evolves according to 
$\opH'$. Nevertheless,
\begin{equation*}
    \Psi'_{S}(\lambda)
	= \me^{-\mi\opH'\lambda}\Psi'_{H}
	= \me^{-\mi\opH'\lambda}\opG\Psi_{H}
	= \opG\me^{-\mi\opH\lambda}\Psi_{H}
	= \opG\Psi_{S}(\lambda) \,,
\end{equation*}
since $\opH'\opG = \opG\opH$.

Under the unitary transformation $\opG$, an operator $\opA$ maps to
$\opA' = \opG \opA \opG^{-1}$.  Now, in the Heisenberg picture, $\opA$
evolves according to the Hamiltonian $\opH$, $\opA_{H}(\lambda) =
\me^{\mi\opH\lambda} \opA_{S} \me^{-\mi\opH\lambda}$, while $\opA'$
evolves according to $\opH'$, $\opA'_{H}(\lambda) =
\me^{\mi\opH'\lambda} \opA'_{S} \me^{-\mi\opH'\lambda}$.  However,
since $\opA'$ is also an operator on $\HilbH$, it can equally well be
evolved using $\opH$.  This gives the \emph{interaction picture} form
for $\opA'$:
\begin{equation*}
    \opA'_{I}(\lambda) 
	= \me^{\mi\opH\lambda} \opA'_{S} \me^{-\mi\opH\lambda}
        = \opG(\lambda) \opA_{H}(\lambda) \opG^{-1}(\lambda) \,,
\end{equation*}
where $\opG(\lambda) = \me^{\mi\opH\lambda} \opG
\me^{-\mi\opH\lambda}$.

The corresponding interaction-picture form for a state $\Psi'$ must
then be related to its Schr\"odinger-picture form by
\begin{equation*}
    \Psi'_{I}(\lambda) 
	= \me^{\mi\opH\lambda} \Psi'_{S}(\lambda) 
	= \opG(\lambda) \Psi_{H}\,,
\end{equation*}
so that
\begin{equation*}
    \bra{\Psi'_{I}(\lambda)}\opA'_{I}(\lambda)\ket{\Phi'_{I}(\lambda)}
	= \bra{\Psi'_{S}(\lambda)}\opA'_{S}\ket{\Phi'_{S}(\lambda)}
	= \bra{\Psi'_{H}}\opA'_{H}(\lambda)\ket{\Phi'_{H}} \,,
\end{equation*}
for all $\Psi'$, $\Phi'$ and $\lambda$.  This also allows
Heisenberg-picture states constructed using the field $\oppsi$ to be
compared to interaction-picture states constructed using the field
$\oppsi'$, such that
\begin{equation*}
    \inner{\op{\Psi}_{H}}{\op{\Phi}'_{I}(\lambda)}
	= \bra{\op{\Psi}_{H}} \opG(\lambda) \ket{\op{\Phi}_{H}} 
	= \bra{\op{\Psi}_{S}(\lambda)}\opG\ket{\op{\Phi}_{S}(\lambda)} 
	= \inner{\opPsi_{S}(\lambda)}{\op{\Phi}'_{S}(\lambda)} \,,
\end{equation*}
for any $\lambda$.

To see why the name ``interaction picture'' is applicable here, 
consider the interaction-picture form of the field $\oppsi'[f]$:
\begin{equation*}
    \param{\oppsi'_{I}}[f\lambda]
	= \me^{\mi\opH\lambda} \oppsi'[f] \me^{-\mi\opH\lambda}
	= \opG(\lambda) \param{\oppsi}[f\lambda] 
	  \opG^{-1}(\lambda) \,.
\end{equation*}
Then
\begin{equation} \label{eqn:B5}
    \begin{split}
	\pderiv{\param{\oppsi'_{I}}[f\lambda]}{\lambda}
	    &= \pderiv{\opG(\lambda)}{\lambda} \param{\oppsi}[f\lambda] 
	       \opG^{-1}(\lambda)
	       + \opG(\lambda) 
	         \pderiv{\param{\oppsi}[f\lambda]}{\lambda} 
	         \opG^{-1}(\lambda) 
	       - \opG(\lambda) \param{\oppsi}[f\lambda] 
	         \opG^{-1}(\lambda)
	         \pderiv{\opG(\lambda)}{\lambda}\opG^{-1}(\lambda) \\
	    &= \pderiv{\opG(\lambda)}{\lambda} \opG^{-1}(\lambda) 
	       \param{\oppsi'_{I}}[f\lambda] 
	       + \mi\opG(\lambda)[\opH, \param{\oppsi}[f\lambda]] 
	         \opG^{-1}(\lambda)
	       - \param{\oppsi'_{I}}[f\lambda]
	         \pderiv{\opG(\lambda)}{\lambda}\opG^{-1}(\lambda) \\
	    &= \mi[\opH'(\lambda), \param{\oppsi'_{I}}[f\lambda]]
	       + [\pderiv{\opG(\lambda)}{\lambda}\opG^{-1}(\lambda),
	          \param{\oppsi'_{I}}[f\lambda]] \,,
    \end{split}
\end{equation}
where 
\begin{equation*}
    \opH'(\lambda) 
	= \me^{\mi\opH\lambda} \opH' \me^{-\mi\opH\lambda}
	= \opG(\lambda) \opH \opG^{-1}(\lambda) \,.
\end{equation*}
But, since $\param{\oppsi'_{I}}[f\lambda]$ evolves according to $\opH$,
\begin{equation*}
    \pderiv{\param{\oppsi'_{I}}[f\lambda]}{\lambda}
	= \mi[\opH, \param{\oppsi'_{I}}[f\lambda]] \,,
\end{equation*}
so we can take
\begin{equation*}
    \opH'(\lambda) = \opH + \Delta \opH(\lambda) \,,
\end{equation*}
where
\begin{equation} \label{eqn:B6}
    \Delta\opH(\lambda) 
	= \mi\pderiv{\opG(\lambda)}{\lambda}\opG^{-1}(\lambda) \,.
\end{equation}

Now, suppose that $\oppsi[f]$ is a free field, as in \sect{sect:free},
and that $\opH$ is the free, relativistic Hamiltonian such that
\begin{equation*}
    [\param{\oppsi}[f\lambda], \opH] 
	= (\opP^{2} + m^{2})\param{\oppsi}[f\lambda] \,,
\end{equation*}
where $m$ is the free-particle mass. Then the field equation for 
$\param{\oppsi}[f\lambda]$ is
\begin{equation*}
    \mi \pderiv{\param{\oppsi}[f\lambda]}{\lambda} 
	= (\opP^{2} + m^{2})\param{\oppsi}[f\lambda] \,.
\end{equation*}
However, since $\opG$ and $\opH$ commute with $\opP$, it is also the 
case that
\begin{equation*}
    [\param{\oppsi'_{I}}[f\lambda], \opH'(\lambda)]
	= (\opP^{2} + m^{2})\param{\oppsi'_{I}}[f\lambda] \,.
\end{equation*}
Substituting this into \eqn{eqn:B5} then gives
\begin{equation*}
    \mi\pderiv{\param{\oppsi'_{I}}[f\lambda]}{\lambda}
	= (\opP^{2} + m^{2})\param{\oppsi'_{I}}[f\lambda]
	  + [\Delta\opH(\lambda), \param{\oppsi'_{I}}[f\lambda]] \,.
\end{equation*}
If we take this as the field equation for $\param{\oppsi'_{I}}
[f\lambda]$, then we can consider $\oppsi'_{I}$ to be an
\emph{interacting} field, with the interaction Hamiltonian
$\Delta\opH(\lambda)$.

In traditional QFT, Haag's Theorem essentially disallows the use of
the interaction picture to relate an interacting field to a free
field, stating that any field unitarily related to a free field must
itself be a free field \cite{haag55,hall57,streater64}.  The
traditional interaction picture is relative to evolution in time, and
any two fields unitarily equivalent at any one time will necessarily
have the same vacuum expectation values at that time.  Haag's Theorem
is then proved by showing that Lorentz invariance requires that fields
with equal same-time vacuum expectation values will also have equal
different-time vacuum expectation values and that this implies that if
one of the fields is free the other one must be too.

For parameterized QFT, however, the evolution parameter is not one of
the spacetime coordinates and, thus, parameter evolution is not
related to spacetime transformations.  It is therefore possible to
have a field such as $\param{\oppsi'_{I}}[f\lambda]$ that is unitarily
related to a free field $\param{\oppsi}[f\lambda]$ for each value of
$\lambda$ but with a \emph{different} transformation $\opG(\lambda)$
for each $\lambda$.  In traditional QFT, Lorentz covariance prohibits
having an analogous unitary transformation that depends on time. (For 
a more detailed discussion of Haag's Theorem and parameterized QFT, 
see \refcite{seidewitz15}.)

Note also that the $\lambda$-dependent transformation $\opG(\lambda)$
essentially induces a different effective vacuum state $\ketl{0'} =
\opG(\lambda)\ketvac = \me^{\mi\opH\lambda}\ket{0'}$ for each value of
$\lambda$.  The two-point equal-$\lambda$ vacuum expectation value for
$\param{\oppsi'_{I}} [f\lambda]$, for example, is thus
\begin{equation*}
    \bral{0'}\param{\oppsi'_{I}}[f\lambda]
             \param{\oppsi'_{I}{}\adj}[g\lambda] \ketl{0'}
	= \bravac\param{\oppsi}[f\lambda]
	         \param{\oppsit}[g\lambda] \ketvac \,.
\end{equation*}
When constructing the vacuum expectation values for
$\param{\oppsi'_{I}} [f\lambda]$ across different $\lambda$ values,
though, one must be careful to use the appropriate effective vacuum
for each value.  For example, the two point different-$\lambda$
expectation value is
\begin{equation*}
    \bralp{0'}\param{\oppsi'_{I}}[f\lambda']
             \param{\oppsi'_{I}{}\adj}[g\lambda] \ketl{0'}
	= \bravac\param{\oppsi}[f\lambda'] \opG^{-1}(\lambda')
	         \opG(\lambda) \param{\oppsit}[g\lambda] \ketvac
	\ne \bravac\param{\oppsi}[f\lambda']
	         \param{\oppsit}[g\lambda] \ketvac \,.
\end{equation*}

The effective vacuum $\ketl{0'}$ is a null eigenstate of the effective
Hamiltonian $\opH'(\lambda)$ for each specific value of $\lambda$.
However, it is \emph{not} an eigenstate of $\opH$ for any $\lambda$.
Considered as a function of $\lambda$, $\ketl{0'}$ is therefore an
example of a state that is invariant under Poincar\'e transformations,
but evolves in $\lambda$, as permitted in the parameterized theory.

\subsection{Integrated Fields} \label{sect:integrated}

In general, the two-point vacuum expectation value
\begin{equation} \label{eqn:B7}
    \kersym[f, g; \lambda - \lambdaz] \equiv
	\bravac \param{\oppsi}[f\lambda] 
	       \param{\oppsit}[g\lambdaz] \ketvac
\end{equation}
represents the \emph{propagation} of a particle from parameter value
$\lambdaz$ to $\lambda$.  The quantity $\kersym[f, g; \lambda -
\lambdaz]$ clearly depends only on the difference $\lambda -
\lambdaz$, since applying the same unitary parameter translation to
the field operator at each parameter value leaves the vacuum
expectation value unchanged.  Intuitively, such an expectation value
represents the propagation of a particle over any path of intrinsic
length $\lambda - \lambdaz$ (as similarly described in the
introductory discussion in \sect{sect:why-parameterized}).
Integrating over all intrinsic lengths then gives the amplitude for
propagation over all possible paths:
\begin{equation*}
    \propsym[f, g] 
	\equiv \int_{0}^{\infty} \dl\, \kersym[f, g; \lambda] \,.
\end{equation*}

For this purpose, it is convenient to introduce the \emph{integrated}
particle fields
\begin{equation} \label{eqn:B8}
    \oppsill[f] \equiv \int_{\lambdaz}^{\infty} \dl\, 
			    \param{\oppsi}[f\lambda] \,,
\end{equation}
such that
\begin{equation*}
    \bravac\oppsill[f]\param{\oppsit}[g\lambdaz]\ketvac
	= \bravac \int_{\lambdaz}^{\infty} \dl\, 
	          \param{\oppsi}[f\lambda]
		  \param{\oppsit}[g\lambdaz] \ketvac 
        = \propsym[f,g] \,.
\end{equation*}
for all $\lambdaz$.  Note that the integrated field operator stills
depends on a parameter value $\lambdaz$, even though $\propsym[f,g]$
is independent of it.  This reflects the fact that, even when
integrating over intrinsic path lengths, the freedom remains to
arbitrarily choose the parameter value at the start of the paths being
integrated along.  

Parameter translation is, in fact, a global symmetry of the integrated
field, so such a field has the eleven-dimensional symmetry group
$\GroupP \otimes U(1)$, where $\GroupP$ is the ten-dimensional
(inhomogeneous) Poincar\'e group.  This is just the \emph{Four-Space
Formalism} (FSF) group described in \refcite{fanchi79}.  (However, it
is different than the $\GroupSO(3,2)$ and $\GroupSO(4,1)$ groups
discussed in \refcites{saad89} and \cite{shnerb93}.)

The following theorems follow directly from the definition of the
integrated field and the axioms for the corresponding unintegrated
field.

\begin{theorem} \label{thm:integrated-1}
    For each parameter value $\lambdaz$ and all $f \in \HilbD$, the
    operator $\oppsill[f]$ is defined on a domain $K'$ of states dense
    in $\HilbK'$ containing the vacuum state $\bra{0}$.  It carries
    vectors in $K'$ into vectors in $K'$.
\end{theorem}

The domain $K'$ here is a subset of the generalized state space
$\HilbK'$, rather than $\HilbH$, because the result of applying
$\oppsill[f]$ to a state in $\HilbH$ generally results in a vector
that is not normalizable.  As a result, an expression such as
$\bra{\Psi}\oppsill[f]\ket{\Phi}$ should strictly be understood as the
application of $\bra{\Psi}\oppsill[f] \in \HilbK'$ to $\ket{\Phi} \in
\HilbK$.  However, with a further common abuse of ket notation, we can
formally use the notation $\opA\ket{\Phi}$ for any $\opA$ constructed
as a polynomial in the integrated fields, with the understanding that
this generally only has meaning in the context of an expression like
$\bra{\Psi}\opA\ket{\Phi}$, in which case it means $\bra{\Psi}\opA$
applied to $\ket{\Phi}$.

\begin{theorem} \label{thm:integrated-2}
    For any Poincar\'e transformation $\{\Delta x, \Lambda\}$, 
    \begin{equation*}
	\opU(\Delta x, \Lambda) \oppsill[f] 
	\opU^{-1}(\Delta x, \Lambda)
	    = \oppsill[\{\Delta x, \Lambda\}f] \,.
    \end{equation*}
    Given the Hamiltonian $\opH$, for any parameter translation
    $\Delta\lambda$,
    \begin{equation*}
	\me^{\mi\opH\Delta\lambda} \oppsill[f]
	\me^{-\mi\opH\Delta\lambda} 
	     = \oppsilll{\lambdaz + \Delta\lambda}[f] \,.
    \end{equation*}
    These equations are valid when applied to any element of $K'$.
\end{theorem}

The use of the integrated fields can be generalized to cover any
$\opA$ in the algebra $\HilbA$ of operators formed as polynomials in
the (unintegrated) smeared field operators.  Given the functional form
for $\opA$ in \eqn{eqn:B3}, define
\begin{equation*}
    \opbar{A}[\{f_{i}\}, \{g_{j}\}] 
	\equiv \opA(\{\oppsill[f_{i}]\}, 
	       \{\param{\oppsit}[{g_{j}}\lambdaz]\})
	     = \opA[\{\bar{f}_{i}\}, \{g_{j}\}] \,.
\end{equation*}
That is, all instances of the $\oppsi$ operators in the construction
of $\opA$ are replaced with their integrated versions, but \emph{not}
instances of the adjoint $\oppsit$ operators (for which the
Heisenberg-picture versions are used here, with a parameter argument
consistent with the integrated fields). To simplify the notation, an 
overbar will also be placed on a functional argument to indicate that 
it is to be used as the argument of an integrated field, as in 
$\bar{f}_{i}$ in the second equality above.

Define the \emph{special adjoint} $\opbar{A}\dadj$, such that
\begin{equation} \label{eqn:B9}
    \opbar{A}\dadj[\{g_{j}\}, \{f_{i}\}] 
	\equiv \overline{\opA\adj}[\{g_{j}\}, \{f_{i}\}]
	     = \opA\adj[\{\bar{g}_{j}\}, \{f_{j}\}]
	     = (\opA[\{f_{i}\}, \{\bar{g}_{j}\}])\adj\,,
\end{equation}
where $\opA\adj$ is defined as in \eqn{eqn:B4}.  Note that the
special adjoints differs from the regular adjoint of $\opbar{A}$ in
that
\begin{equation*}
    \opbar{A}\adj[\{g_{j}\}, \{f_{i}\}]
	= (\opbar{A}[\{f_{i}\}, \{g_{j}\}])\adj
	= \opA\adj[\{g_{j}\}, \{\bar{f}_{i}\}]
        \neq \opbar{A}\dadj[\{g_{j}\}, \{f_{i}\}] \,.
\end{equation*}
In particular, the special adjoint of the field operator
$\oppsill[f]$ is
\begin{equation*}
    \oppsill\dadj[f] 
	= \oppsill\dadj[\{\}, \{f\}]
	= (\param{\oppsi}[{\{f\}, \{\}}\lambdaz])\adj
	= \param{\oppsi\adj}[f\lambdaz]
\end{equation*}
and
\begin{equation*}
    \param{\oppsi\dadj}[f\lambdaz] = (\oppsill^{\ddag\ddag}[f])\adj 
                                   = \oppsill\adj[f] \,.
\end{equation*}

This special adjoint has all the properties required of an involution
on the *-algebra of operators of the form $\opbar{A}$.  First, 
\begin{equation*}
    \opbar{A}\dadj{}\dadj = (\overline{\opA\adj})\dadj
                          = \overline{\opA\adj{}\adj}
			  = \opbar{A} \,.
\end{equation*}
Further, clearly
\begin{equation*}
    (\opbar{A} + \opbar{B})\dadj = \opbar{A}\dadj + \opbar{B}\dadj \,.
\end{equation*}
and, for any complex number $a$,
\begin{equation*}
    (a\opbar{A})\dadj 
	= (\overline{a\opA})\adj
	= \overline{a\conj\opA\adj}
	= a\conj\overline{\opA\adj}
	= a\conj\opbar{A}\dadj \,.
\end{equation*}
Finally, note that, 
\begin{equation*}
    \begin{split}
	\opbar{A}[\{f_{i}\}, \{g_{j}\}]
	\opbar{B}[\{f'_{i}\}, \{g'_{j}\}]
	&= \opA[\{\bar{f}_{i}\}, \{g_{j}\}]
	   \opB[\{\bar{f}'_{i}\}, \{g'_{j}\}] \\
	&= (\opA\opB)[\{\bar{f}_{i}\}, \{\bar{f}'_{i}\};
	              \{g_{j}\}, \{g'_{j}\}] \\
	&= (\overline{\opA\opB})[\{f_{i}\}, \{f'_{i}\};
	                         \{g_{j}\}, \{g'_{j}\}] \,,
    \end{split}
\end{equation*}
so $\opbar{A}\opbar{B} = \overline{\opA\opB}$. Then
\begin{equation*}
    (\opbar{A}\opbar{B})\dadj 
	= (\overline{\opA\opB})\dadj
	= \overline{(\opA\opB)\adj}
	= \overline{\opB\adj\opA\adj}
	= \overline{\opB\adj}\overline{\opA\adj}
	= \opbar{B}\dadj\opbar{A}\dadj \,.
\end{equation*} 

With the above definitions, the following theorem holds.

\begin{theorem} \label{thm:integrated-3}
    Operators constructed as polynomials in the field $\oppsill[f]$
    and its adjoint $\oppsill\dadj[f]$ form a *-algebra
    $\bar{\HilbA}$, with involution given by the special adjoint
    $\opbar{A}\dadj$ for $\opbar{A} \in \bar{\HilbA}$, isomorphic to
    the algebra $\HilbA$ defined using the unintegrated fields.
\end{theorem}

It is also possible to define unsmeared versions of the integrated 
fields:
\begin{equation*}
    \oppsill[f] = \intfour x\, f\conj(x) \oppsill(x) \,,
\end{equation*}
where the unsmeared integrated fields are
\begin{equation*}
    \oppsill(x) 
	\equiv \int_{\lambdaz}^{\infty} \dl\, \oppsil{x} \,.
\end{equation*}
The special adjoint operates on these fields as $\oppsill\dadj(x) =
\oppsitlz{x} \text{ and } \oppsill^{\ddag\ddag}(x) = \oppsill(x)$.

\subsection{Antiparticles} \label{sect:antiparticles}

As mentioned in \sect{sect:why-off-shell}, in perturbative QFT,
backward-time particle propagation is traditionally re-interpreted as
forward-time antiparticle propagation.  Alternatively, one can also
view antiparticles as particles that formally propagate backwards in
time \cite{feynman49}.  Indeed, Stueckelberg consider segments of
particle paths with reversed-time propagation to represent
antiparticles \cite{stueckelberg41,stueckelberg42}.  In
\refcite{seidewitz06a}, I presented a related but importantly
different formulation, in which only the relative time-ordering of the
end points of a particle path determined whether the propagation is
to be interpreted as for a particle or antiparticle, with the path
otherwise unconstrained between these points.

However, in \refcite{seidewitz11}, I found that, in order to produce
complete Feynman diagrams for scattering using the spacetime path
formalism, it was also necessary to introduce particle propagation
that was reversed in the sense of the path parameter $\lambda$, as
well as propagation reversed in the sense of time.  It turns out,
though, that such reverse-propagating particles can also be used
directly to represent antiparticles, under the condition in which
forward propagation in $\lambda$ is aligned with forward propagation
in time.  This will be discussed further in \sect{sect:scattering}.
In the present section, I will present the formal model of
antiparticles as reverse-propagating particles, which will be used in
the modeling of interactions in \sect{sect:application}.

To do this, first note that, in \eqn{eqn:B8} defining $\oppsill$, the
integral is only over parameter values $\lambda > \lambdaz$.  This
reflects the conception that $\lambda - \lambdaz$ is a path length
and, therefore, must be positive.  However, it was actually an
arbitrary choice to presume that particle propagation along a path was
for values of $\lambda$ \emph{increasing} from $\lambdaz$.  One could
just as well take $\lambda$ to \emph{decrease} from $\lambdaz$ and
consider the (positive) intrinsic path lengths $\lambdaz - \lambda$.

Such \emph{reverse} propagation is reflected in the field vacuum
expectation values, such that
\begin{equation*}
    \bravac \param{\oppsinn{-}}[f\lambda] 
	   \param{\oppsinn{-}\adj}[g\lambdaz] \ketvac
	= \kersym[f, g; \lambdaz - \lambda]\,,
\end{equation*}
where $\lambdaz - \lambda$ is positive for $\lambda < \lambdaz$.  Note
the difference with \eqn{eqn:B7}, even assuming an equivalent
underlying form for the functional $\kersym[f, g; \Delta\lambda]$.
For a \emph{forward}-propagating field, if $\lambda < \lambdaz$, the
$\Delta\lambda$ argument in $\kersym[f, g; \Delta\lambda]$ is
negative, while, for a reverse-propagating field, it is positive.

Therefore, in terms of propagation, a reverse-propagating field really
is a different sort of field than a forward-propagating one.
Nevertheless, for fields that otherwise have vacuum expectation values
given by the same functional $\kersym[f, g; \Delta\lambda]$ (e.g.,
particles of the same mass and spin), we will take the
reverse-propagating field $\oppsinn{-}$ to represent the
\emph{antiparticle} of the particle represented by the corresponding
forward-propagating field $\oppsinn{+}$ (as motivated earlier).  That
is, antiparticles propagate in the \emph{opposite direction} in
$\lambda$ from their corresponding particles.

This also means that, when defining integrated fields for
antiparticles, for a given $\lambdaz$, propagation is \emph{downwards}
from $\lambdaz$ to lower values of $\lambda$.  Therefore, the
integration must be over all $\lambda < \lambdaz$:
\begin{equation*}
    \oppsinnll{-}[f] 
	\equiv \int_{-\infty}^{\lambdaz} \dl\, 
	    \param{\oppsinn{-}}[f\lambda] \\
\end{equation*}
However, the integrated vacuum expectation value is then
\begin{equation*}
    \begin{split}
	\bravac\oppsinnll{-}[f]\oppsinnll{-}\dadj\ketvac 
	    &= \bravac \int_{-\infty}^{\lambdaz} \dl\, 
		\param{\oppsinn{-}}[f\lambda] 
		\param{\oppsinn{-}\adj}[g\lambdaz] \ketvac
	    = \int_{-\infty}^{\lambdaz} \dl\, 
		\kersym[f, g; \lambdaz - \lambda] \\
 	    &= \int_{-\infty}^{0} \dl\, 
 		\kersym[f, g; -\lambda]
	    = \int^{\infty}_{0} \dl\, 
		\kersym[f, g; \lambda]
	    = \propsym[f, g] \,,
	\end{split}
\end{equation*}
giving, in the end, the same propagation factor as for regular 
particles. All the theorems given in \sect{sect:integrated} apply 
equally to the reverse-propagated integrated field, as do the 
definitions of its corresponding *-algebra and special adjoint.

\section{Application} \label{sect:application}

This section considers the application of the theory presented in
\sect{sect:theory} to the task of modeling interactions and computing
scattering amplitudes.  In particular, the goal is to construct a
unitary transformation $\opG(\lambda)$ between a set of free and
interacting fields of different types of particles, representing a
realistic form of interaction, and then to describe how scattering
amplitudes can be computed from this.

\Sect{sect:states} defines the formalism of particle states to be used
in this section.  \Sect{sect:interactions} then describes how to
construct $\opG(\lambda)$ from particle interactions and
\sect{sect:regularization} briefly addresses the problem of
regularization.  Finally, \sect{sect:scattering} uses this
construction to show how the traditional Dyson expansion for
scattering amplitudes can be derived in the parameterized theory, but
without the mathematical inconsistency associated with the traditional
derivation.

\subsection{Particle States} \label {sect:states}

For the purposes of \sect{sect:application}, it will generally be
convenient to use the unsmeared form for field operators.  Let
$\HilbH$ be the free-particle Fock space and $\oppsi(x)$ be a free
field defined on it, as described in \sect{sect:free}.  Take
$\opH_{0}$ to be the free, relativistic Hamiltonian such that
\begin{equation*}
    \mi\pderiv{\oppsixl}{\lambda} 
	= [\oppsixl, \opH_{0}]
	= (\opP^{2} + m^{2})\oppsixl
	= (-\frac{\partial^{2}}{\partial x^{2}} + m^{2}) \oppsixl \,,
\end{equation*}
the solution of which is
\begin{equation} \label{eqn:C1}
    \oppsixl = \intfour \xz\, \kerneld \oppsilz{\xz} \,,
\end{equation}
where $\kerneld$ is the propagation kernel given in \eqn{eqn:A3}.

For any four-position $x$, define the \emph{position-state} bras
$\brax \equiv \bravac\oppsi(x) \in \HilbK'$.  These bras act as a
continuous basis for $\HilbK$, such that, for any $\ket{f} \in
\HilbK$, $\inner{x}{f} = f(x)$.  Informally, we can also define the
dual kets $\ketx$, such that $\inner{x'}{x} = \delta^{4}(x' - x)$, so
\begin{equation} \label{eqn:C2}
    \ket{f} = \intfour x\, f(x) \ketx \,.
\end{equation}
Of course, with the above normalization, the $\ketx$ are not actually
in $\HilbK$ (or $\HilbH$), but, similarly to the unsmeared field
operators, they may be used with care as long as it is understood that
it is really only expressions such as \eqn{eqn:C2} that are properly
defined, for $f \in \HilbD$.

For example, consider the momentum version of such single particle
states, $\ketp = \oppsit(p)\ketvac$, where $\oppsi(x)$ is the
four-dimensional Fourier transform of $\oppsi(x)$.  The free
Hamiltonian acts on these momentum states as $\opH\ket{p} = (p^{2} +
m^{2})\ket{p}$.  So, it would seem that, for an on-shell state with
$p^{2} = -m^{2}$, one would have $\opH\ket{p} = 0$, which would be a
non-vacuum null eigenstate of $\opH$, in violation of the requirements
for a Hamiltonian operator.  

There is, of course, no actual violation here, since $\ketp \notin
\HilbH$.  On the other hand, $\bra{p} \in \HilbK'$ \emph{is} properly
defined, for all $p$, such that $\inner{p}{f} = f(p)$, where $f(p)$ is
the Fourier transform of $f(x)$.  One can thus define
\begin{equation*}
    \bra{m} \equiv \intfour p\, \delta^{4}(p^{2} + m^{2}) \theta(p^{0}) 
                   \bra{p}
	         = \int \frac{\dthree p}{2\Ep}\, \bra{\Ep, \threep}
\end{equation*}
such that $\inner{m}{f}$ as the probability amplitude for an arbitary
state $\ket{f}$ to be on shell.  (A similar approach is effectively
used in \sect{sect:scattering} to ensure the on-shell nature of
external legs of a scattering process.)

With the above caveat, then, we can take $\ketx = \oppsit(x)\ketvac$.
Similarly, for the Heisenberg-picture fields, we can define $\ketxl
\equiv \oppsitl{x}\ketvac$.  Now, despite their apparent dependence on
$\lambda$, the $\ketxl$ are actually Heisenberg-picture states, since
the $\oppsitl{x}$ are Heisenberg-picture operators.  They represent
the state of a particle localized at a specific four-position $x$
\emph{and} at a specific parameter value $\lambda$.

Multiple applications of $\oppsitl{x}$ result in the multi-particle 
position states
\begin{equation*}
    \ket{\xliN} \equiv 
	\param{\oppsit}({x_{1}}{\lambda_{1}})\cdots
	\param{\oppsit}({x_{N}}{\lambda_{N}})\ketvac \,,
\end{equation*}
which represent multiple particles localized at the four-positions 
$x_{i}$, each with their own specific parameter values $\lambda_{i}$. 
It will also be convenient to have a shorthand notation for the case 
of multiple particles at different positions, but all with the same 
parameter value $\lambdaz$:
\begin{equation*}
    \ket{\xlziN} \equiv \ket{\seqN{\xlzi}} \,.
\end{equation*}

Clearly, for any one value of $\lambda$,
\begin{equation*}
    \innerll{x}{\xz} = \bravac\oppsil{x}\oppsitl{\xz}\ketvac
                    = \delta^{4}(x - \xz) \,,
\end{equation*}
by \axm{axm:commutation} and the normalization of $\ketvac$. However, 
for different $\lambda$ values, using \eqn{eqn:C1} gives
\begin{equation*}
    \innerllz{x}{\xz} = \bravac\oppsil{x}\oppsitlz{\xz}\ketvac
                     = \kerneld \,.
\end{equation*}
The two-point different-$\lambda$ expectation value for $\oppsixl$ is
the probability amplitude for the propagation of a particle from
position $\xz$ at parameter value $\lambdaz$ to position $x$ at
parameter value $\lambda$.  As discussed in \sect{sect:integrated},
the full propagation amplitude is then given by integrating over
$\lambda$, which is equal to the vacuum expectation value of the
integrated field:
\begin{equation} \label{eqn:C3}
    \bravac  \oppsill(x)\oppsill\dadj(\xz) \ketvac
	= \int_{\lambdaz}^{\infty} \dl\, \innerllz{x}{\xz}
	= \int_{\lambdaz}^{\infty} \dl\, \kerneld
	= \prop \,.
\end{equation}
(Note the use here of the special adjoint defined in
\sect{sect:integrated}.)

Next, following on the idea introduced at the end of
\sect{sect:expectation}, consider an interacting-field operator
$\param{\oppsi\Int}(x\lambda)$ that, in the interaction picture, is
related to the free-field operator by a $\lambda$-dependent unitary
transformation:
\begin{equation*}
    \param{\oppsi\Int}(x\lambda) = 
	\opG(\lambda)\oppsixl\opG^{-1}(\lambda) \,,
\end{equation*}
where $\opG(\lambda)$ is constructed as a functional of the field
operator.  Actually, it will be more convenient to work primarily with
integrated field operators.  However, rather than directly integrating
$\param{\oppsi\Int}(x\lambda)$ (which would not be proper in the
interaction picture), instead take $\oppsill\Int(x)$ to be the
operator in the algebra $\bar{\HilbA}$ corresponding to
$\param{\oppsi\Int}(x\lambda)$ in the algebra $\HilbA$ (as discussed
at the end of \sect{sect:integrated}).  This is
\begin{equation*}
    \oppsill\Int(x) = 
	\opbar{G}(\lambdaz)\oppsill(x)\opbar{G}^{-1}(\lambdaz) \,,
\end{equation*}
where the operator $\opbar{G}$ is has the same functional form as
$\opG$, but with the unintegrated field operator and its adjoint
replaced by the integrated field operator and its special adjoint.
With this notation, it is easy to create interacting-particle
position states:
\begin{equation} \label{eqn:C4}
    \begin{split}
	\ket{\xlziN}\Int 
	     &\equiv \oppsill\Int{}\dadj(x_{1})\cdots
		     \oppsill\Int{}\dadj(x_{N})\ketlz{0}\Int \\
	     &= \opbar{G}(\lambdaz)
	             \oppsill\dadj(x_{1})\cdots
		     \oppsill\dadj(x_{N})\ketvac \\
	     &= \opbar{G}(\lambdaz)
		     \param{\oppsit}({x_{1}}{\lambdaz})\cdots
	             \param{\oppsit}({x_{N}}{\lambdaz})\ketvac \\
	     &= \opbar{G}(\lambdaz)\ket{\xlziN} \,,
     \end{split}
\end{equation}
where $\ketlz{0}\Int \equiv \opbar{G}(\lambdaz)\ketvac$ is the
\emph{interacting vacuum}.

\subsection{Interactions} \label{sect:interactions}

The goal now is to construct the \emph{interaction operator}
$\opbar{G}(\lambdaz)$ to represent interactions between different
types of particles.  In order to do this, introduce a set of
(integrated) free fields $\oppsinll(x)$, indexed by the \emph{particle
type} $n$.  Each of these fields act on the same Hilbert space
$\HilbH$, with the same domain $D$, they each individually satisfy the
axioms from \sect{sect:fields} and each otherwise commutes with all
fields of other particle types (and their adjoints).  They also each
have corresponding interacting fields given by
\begin{equation*}
    \oppsinll\Int(x) = 
	\opbar{G}(\lambdaz)\oppsinll(x)\opbar{G}^{-1}(\lambdaz) \,.
\end{equation*}
While the free field $\oppsinll$ of a specific particle type $n$ is
independent of the fields for all other particle types, the 
interacting field $\oppsinll\Int$ will depend on the fields for other 
particle types through the construction of $\opbar{G}$.

Note that, if $\opG$ is unitary, then, from the definition of the 
special adjoint, \eqn{eqn:B9},
\begin{equation*}
    \opbar{G}\dadj = \overline{\opG\adj}
                   = \overline{\opG^{-1}}
		   = \opbar{G}^{-1} \,.
\end{equation*}
That is, $\opbar{G}$ is unitary with respect to the special adjoint.
From now forward, we will drop the overbar from $\opbar{G}$ and just
consider the operator $\opG$ as constructed from the integrated field
operators and their adjoints.  Since $\opG$ is unitary with respect to
the special adjoint, we can take it to have the form
\begin{equation} \label{eqn:C4a}
    \opG(\lambdaz) = \me^{-\mi \opV(\lambdaz)} \,,
\end{equation}
where $\opV(\lambdaz)$ is self-adjoint with respect to the special
adjoint (that is, $\opV(\lambdaz) = \opV\dadj(\lambdaz)$).  

From \thm{thm:integrated-2}, an integrated field $\oppsill(x)$ evolves
in $\lambdaz$ according to the Hamiltonian of the corresponding
unintegrated field $\oppsixl$.  Therefore, we can carry the discussion
of interacting fields in \sect{sect:interacting} over to the case of
integrated fields, as currently being applied.  Thus, as described at
the end of \sect{sect:interacting}, take the effective
interaction-picture Hamiltonian for the interacting field
$\oppsill\Int(x)$ to have the form
\begin{equation*}
    \opH(\lambdaz) = \opH_{0} + \opH\Int(\lambdaz) \,,
\end{equation*}
where $\opH\Int(\lambdaz)$ is the interaction Hamiltonian in the
interaction picture. Then, using \eqns{eqn:B6} and \eqref{eqn:C4a},
\begin{equation} \label{eqn:C5}
    \opH\Int(\lambdaz) 
	 = -\mi \deriv{\opG(\lambdaz)}{\lambdaz}\opG^{-1}(\lambdaz)
	 = -\deriv{\opV(\lambdaz)}{\lambdaz} \,.
\end{equation} 

Further, taking the series expansion for $\opG$,
\begin{equation} \label{eqn:C6}
    \opG(\lambdaz) 
	= \sum_{m} \frac{(-\mi)^{m}}{m!} \opV^{m}(\lambdaz) \,,
\end{equation}
the $m$th term in the series represents the case of exactly $m$
interactions.  The factor $1/m!$ accounts for the $m!$ possible
permutations of the $m$ factors of $\opV$ in the term.  Thus, $\opV$
is the \emph{vertex operator} that determines the effect of the
individual interaction vertices collected within $\opG$. Take
\begin{equation} \label{eqn:C7}
    \opV(\lambdaz) = \intfour x\, \opV_{\lambdaz}(x) \,,
\end{equation}
for some $\opV_{\lambdaz}(x)$ that is self-adjoint.
$\opV_{\lambdaz}(x)$ represents an interaction at the specific
four-position $x$ and is to be constructed from the integrated field
operators.  

Now, we need to account for both particles and antiparticles when
constructing the interaction at a vertex.  Let $\na$ denote a particle
type and $\nr$ denote the corresponding antiparticle type.  Then, as
discussed in \sect{sect:antiparticles}, we take the particle field
$\oppsinaxl$ to be forward propagating, while the antiparticle field
$\oppsinrxl$ is backward propagating.  Therefore,
\begin{equation*}
    \bravac\oppsinarxl\oppsinartlz{\xz}\ketvac = \kernelnar \,,
\end{equation*}
where $\kerneln$ is $\kerneld$ from \eqn{eqn:A3}, using the mass
$m_{n}$, which is the mass of both particles of type $\na$ and the
corresponding antiparticles of type $\nr$.  Note also that a particle
field is taken to commute with the corresponding antiparticle field
(and its adjoint), as for fields of completely different particle
types.  The integrated fields $\oppsinarll(x)$ are constructed as
described in \sect{sect:antiparticles}, such that
\begin{equation*}
    \bravac\oppsinarll(x)\oppsinarll\dadj(\xz)\ketvac = \prop \,.
\end{equation*}

Then, at a specific interaction position $x$, there may be either the
destruction of an ``incoming'' particle by $\oppsinall(x)$ or the
creation of an ``outgoing'' antiparticle by $\oppsinrll\dadj(x)$.
This can be reflected in the combined field \cite{seidewitz11}
\begin{equation*}
    \oppsinll(x) \equiv \oppsinall(x) + \oppsinrll\dadj(x) \,.
\end{equation*}
Since the special adjoint has the properties of an involution, we also
have
\begin{equation*}
    \oppsinll\dadj(x) \equiv \oppsinall\dadj(x) + \oppsinrll(x) \,,
\end{equation*}
which represents the creation of an ``outgoing'' particle by
$\oppsinall\dadj(x)$ or the destruction of an ``incoming''
antiparticle by $\oppsinrll(x)$.

An interaction vertex can be represented in terms of these operators
as
\begin{equation} \label{eqn:C8}
    \opV_{\lambdaz}(x) = 
	g :\prod_{i} \oppsinnll{n_{i}}\dadj(x) \oppsinnll{n_{i}}(x)
	         \prod_{j} \oppsinnll{n_{j}}'(x): \,,
\end{equation}
where $g$ is a coupling constant, $:\cdots:$ represents normal
ordering (that is, placing all $\oppsi\dadj$ operators to the left of
all $\oppsi$ operators in any product, for particles or
antiparticles), and the $\oppsi'$ are self-adjoint fields given by
\begin{equation*}
    \oppsinnll{n_{j}}'(x) 
	\equiv \oppsinnll{n_{j}}(x) + \oppsinnll{n_{j}}\dadj(x) \,.
\end{equation*}
A $\opV_{\lambdaz}(x)$ constructed in this way is clearly
self-adjoint, as required.  Further, it has the important commutivity
property given in the following theorem (proved in the appendix).

\begin{theorem}[Commutivity of the Vertex Operator] \label{thm:vertex}
    Let $\opV_{\lambda}(x)$ be defined as in \eqn{eqn:C8}.  Then
    \begin{equation*}
	[\opV_{\lambda_{1}}(x_{1}), \opV_{\lambda_{2}}(x_{2})] = 0 \,,
    \end{equation*}
    for all values of the $x_{i}$ and $\lambda_{i}$.
\end{theorem}

\begin{corollary*}
    Let
    \begin{equation*}
	\opV(\lambda) = \intfour x\, \opV_{\lambda}(x) \,.
    \end{equation*}
    Then
    \begin{equation*}
	[\opV(\lambda_{1}), \opV(\lambda_{2})] = 0 \,,
    \end{equation*}
    for all values of the $\lambda_{i}$.
\end{corollary*}

Because the $\opV(\lambda)$ commute for different $\lambda$, so will
$\opH\Int(\lambda) = -\dif \opV(\lambda) / \dl$.  Therefore,
\eqn{eqn:C5} may be integrated to obtain
\begin{equation*}
    \opV(\lambdaz) =  
	\int_{\lambdaz}^{\infty} \dl\, \opH\Int(\lambda) \,.
\end{equation*}
Note that, for the chosen integration bounds and any non-trivial
$\opH\Int(\lambdaz)$, $\opV(\lambdaz)$ will be non-zero for all finite
$\lambdaz$, and so $\opG(\lambdaz)$ will not be the identity for any
$\lambdaz$.  Thus, unlike the traditional interaction picture, in
which the free and interacting fields coincide at one point in time,
there is no $\lambdaz$ for which the free and interacting
parameterized fields are ``the same.''

\subsection{Regularization} \label{sect:regularization}

The operator $\opV_{\lambdaz}(x)$, as defined using normal ordering in
\eqn{eqn:C8}, is actually a proper distribution.  However, this is no
longer the case when it is integrated over all spacetime, as in
\eqn{eqn:C7}.  As a result, the series expansion of $\opG$ in
\eqn{eqn:C6} gives terms with integrals (over spacetime or momentum
space) that are not well defined.  These integrals require
regularization, as in the case of perturbation expansions in
traditional QFT.

As an example, consider the use of Pauli-Villars regularization 
\cite{pauli49}. In this approach, the Feynman propagator for a 
particle of mass $m$,
\begin{equation*}
    \propsym(x-\xz;m) 
	= -\mi(2\pi)^{-4}\intfour p\,
		\frac{\me^{\mi p\cdot(x - \xz)}}
		     {p^{2} + m^{2} - \mi\varepsilon} \,,
\end{equation*}
is replaced with the effective propagator $\propsym(x-\xz;m) -
\propsym(x-\xz;M)$, for $M > m$.  The integrals in the expansion are
regulated by the extra term in $M$ and can be evaluated.  Then, at
the end of the calculation, $M$ is taken to infinity, such that
$\propsym(x-\xz;M) \rightarrow 0$.

Now, consider that (see \eqn{eqn:C3})
\begin{equation*}
    \propsym(x;m)
	= \int_{\lambdaz}^{\infty} \dl\, \intfour p\, 
	    \me^{\mi p\cdot(x - \xz)}
	    \me^{-\mi(p^{2} + m^{2})(\lambda - \lambdaz)}
	= \int_{\lambdaz}^{\infty} \dl\,
	    [\oppsixl, \oppsitlz{\xz}] \,.
\end{equation*}
Therefore,
\begin{equation*}
    \begin{split}
    	\propsym(x-\xz;m) - \propsym(x-\xz;M) 
	    &= \int_{\lambdaz}^{\infty} \dl\, \intfour p\, 
		\me^{\mi p\cdot(x - \xz)}
		\left[ 
		    \me^{-\mi(p^{2} + m^{2})(\lambda - \lambdaz)} -
		    \me^{-\mi(p^{2} + M^{2})(\lambda - \lambdaz)}
		\right] \\
	    &= \int_{\lambdaz}^{\infty} \dl\,  
		\left[
		    1 - \me^{-\mi(M^{2} - m^{2})(\lambda - \lambdaz)}
		\right]
		\intfour p\, \me^{\mi p\cdot(x - \xz)}
		\me^{-\mi(p^{2} + m^{2})(\lambda - \lambdaz)} \\
	    &= \int_{\lambdaz}^{\infty} \dl\,
		\left[
		    1 - \me^{-\mi(M^{2} - m^{2})(\lambda - \lambdaz)}
		\right]
		[\oppsixl, \oppsitlz{\xz}] \,.
    \end{split}
\end{equation*}

Define
\begin{equation*}
    \oppsi_{f}(x) 
	\equiv \int_{\lambdaz}^{\infty} \dl\, 
	       f(\lambda - \lambdaz) \oppsixl \,.
\end{equation*}
This is a generalization of the definition of the integrated field 
operator, which reduces to $\oppsill(x)$ for $f(\lambda - \lambdaz) = 
1$. Taking, instead,
\begin{equation*}
    f(\lambda - \lambdaz)
	 = 1 - \me^{-\mi(M^{2} - m^{2})(\lambda - \lambdaz)}
\end{equation*}
then gives
\begin{equation*}
    [\oppsi_{f}(x),\oppsill\dadj(\xz)]
	= \propsym(x-\xz;m) - \propsym(x-\xz;M) \,.
\end{equation*}

Thus, we can achieve Pauli-Villars regularization simply by replacing
$\oppsinll(x)$ with $\oppsi_{n,f_{n}}(x)$ in the vertex operator
$\opV$, for each particle type $n$ (with $f_{n}$ defined as given
above, but with $m$ replaced by $m_{n}$), but leaving the adjoints
$\oppsinll\dadj(x)$ unchanged.  This is similar to the idea discussed
in Sect.  IIID of \refcite{seidewitz06a}, based on earlier proposals
in \refcites{frastai95,land97,land03}, to regularize by making the
interaction coupling dependent on the intrinsic path length (though in
those cases the function for correlation in the path parameter is more
physically motivated).

The topic of regularization in parameterized QFT, and its relation to 
renormalization of the expansion of $\opG$, will be explored in 
future work. For the present paper, we will simply consider 
\eqn{eqn:C6} to be a formal series expansion and see, in the next 
section, how this expansion can be used to reproduce the similarly 
formal series expansion that results from perturbative scattering 
theory in traditional QFT.

\subsection{Scattering} \label{sect:scattering}

This section shows how the usual scattering amplitudes can be
reproduced by the parameterized formalism, given a vertex operator of
the form defined in \eqn{eqn:C8}.  We will be using only integrated
fields in this section.  Presume that a specific, fixed value of
$\lambdaz$ has been chosen, so the notation may be simplified by
omitting explicit references to $\lambdaz$.

Consider, first, scattering that takes place limited to just a
specific four volume $\Vol$.  A vertex operator $\opV_{\Vol}$
restricted to $\Vol$ may be defined as in \eqn{eqn:C7} but with the
integral over all spacetime replaced by an integral over only the
four-volume $\Vol$.  That is,
\begin{equation*}
    \opV_{\Vol} \equiv \int_{\Vol} \dfour x\, \opV(x) \,.
\end{equation*}
The corresponding interaction operator restricted to $\Vol$ is then
\begin{equation*}
    \opG_{\Vol} \equiv \me^{-\mi \opV_{\Vol}} \,.
\end{equation*}
It is another immediate corollary of \thm{thm:vertex} that the
$\opV(x)$ commute for different $x$, so the $\opV_{\Vol}$ commute for
different $\Vol$.  Therefore, the restricted interaction operator has
the property
\begin{equation*}
    \opG_{\Vol_{1} \cup \Vol_{2}} = \opG_{\Vol_{1}}\opG_{\Vol_{2}} \,,
\end{equation*}
which allows for easy separation of interactions within a system in a 
certain four-volume from interactions that occur in the environment 
of the system \cite{seidewitz11}.

Now consider states $\ket{\Phi\In}$ that are superpositions of
position states with positions outside of $\Vol$, that is, they are
constructed from applications of the field operators
$\oppsina\dadj(x)$ and $\oppsinr\dadj(x)$ that have $x$ outside of
$\Vol$.  Using \eqn{eqn:C4}, but with the restricted interaction
operator, then gives $\opG_{\Vol}\ket{\Phi\In}$ as the state
representing the particles entering $\Vol$ from outside and
interacting there (or not interacting at all).  Since the free and
interacting Hilbert spaces are the same, we can expand the interacting
state in the free position state basis:
\begin{equation*}
    \opG_{\Vol}\ket{\Phi\In} = 
	\sum_{N=0}^{\infty}  
	\left( \prod_{i=1}^{N} \intfour x_{i} \right)
	\ketlz{x_{1},\dots,x_{N}}
	\bralz{x_{1},\dots,x_{N}}\opG_{\Vol}\ket{\Phi\In} \,.
\end{equation*}

The coefficients $\bra{x_{1},\dots,x_{N}} \opG_{\Vol} \ket{\Phi\In}$
are the probability amplitudes that the incoming particles, after
interacting in $\Vol$, result in $N$ particles at the given positions.
Thus, if we construct states $\bra{\Phi\Out}$ using the integrated
fields $\oppsina(x)$ and $\oppsinr(x)$, with $x$ also outside of
$\Vol$, then
\begin{equation*}
    \bra{\Phi\Out}\opG_{\Vol}\ket{\Phi\In} =
	\sum_{N=0}^{\infty}  
	\left( \prod_{i=1}^{N} \intfour x_{i} \right)
	\inner{\Phi\Out}{x_{1},\dots,x_{N}; \lambdaz}
	\bralz{x_{1},\dots,x_{N}}\opG_{\Vol}\ket{\Phi\In}
\end{equation*}
is the amplitude for the incoming particles to scatter only in $\Vol$
(or not interact at all) and then propagate \emph{out} of $\Vol$ into
the outgoing state $\bra{\Phi\Out}$.

If we take the series expansion for $\opG_{\Vol}$,
\begin{equation*}
    \opG_{\Vol} = \sum_{m}\frac{(-\mi)^{m}}{m!}\opV_{\Vol}^{m} \,,
\end{equation*}
then the $m$th term in the series represents the case of exactly $m$
interactions taking place within $\Vol$.  Using this expansion, any
scattering amplitude $\bra{\Phi\Out} \opG_{\Vol} \ket{\Phi\In}$ can be
written as a (weighted) sum of terms of the form
\begin{equation*}
    \bra{0}\oppsinn{n'_{1}}(x'_{1})\cdots\oppsinn{n'_{N}}(x'_{N})
      \frac{(-\mi)^{m}}{m!}\opV_{\Vol}^{m}
      \oppsinn{n_{1}}\dadj(x_{1})\cdots
	  \oppsinn{n_{M}}\dadj(x_{M})\ket{0} \,,
\end{equation*}
where the $x_{i}$ and $x'_{j}$ are all outside $\Vol$ and each 
$n_{i}$ and $n'_{j}$ may be a particle or antiparticle type.

As in the traditional derivation of Feynman diagrams, we can now move
all annihilation operators to the right in each such term, generating
commutators with intermediate field applications of the same particle
type \cite{peskin95,weinberg95,ticciati99}.  Thus, the pairing of a
factor $\oppsin(x')$ in $\opV_{\Vol}$ with an ``in'' particle factor
$\oppsina\dadj(x)$ gives the commutator
\begin{equation} \label{eqn:C9}
    [\oppsin(x'), \oppsina\dadj(x)] 
	= [\oppsina(x'), \oppsina\dadj(x)]
	= \propsymn(x'-x) \,,
\end{equation}
and the pairing of a factor $\oppsin\dadj(x')$ in $\opV_{\Vol}$ with
an ``in'' antiparticle factor $\oppsinr\dadj(x)$ gives the
commutator
\begin{equation} \label{eqn:C10}
    [\oppsin\dadj(x'), \oppsinr\dadj(x)] 
	= [\oppsinr(x'), \oppsinr\dadj(x)]
	= \propsymn(x'-x) \,.
\end{equation}
Pairings of ``out'' particle and antiparticle annihilation operators
with corresponding creation operators within $\opV$ (or directly with
those for incoming particles) give similar factors, as do the pairings
of annihilation and creation operators within $\opV$.

Next, consider a region $\Vol(t\F,t\I)$ bounded by hyperplanes at $t =
t\I$ and $t = t\F > t\I$, but unbounded in space.  Further, let
$\ket{\Phi\In} = \ket{\Phi\I}$, a superposition of position states
with $t < t\I$, and $\bra{\Phi\Out} = \bra{\Phi\F}$, a superposition
of positions states with $t > t\F$. Then, in the commutators of 
\eqns{eqn:C9} and \eqref{eqn:C10} for incoming particles and 
antiparticles, $t' = x^{\prime 0} \geq t\I > t = x^{0}$, while, in the 
commutators for outgoing particles and antiparticles, $t' > t\F \geq 
t$.

Using \eqn{eqn:A1}, it is clear that, for $t' > t$, $\propsymn(x'-x)
= \propsymn(t'-t, \threexp-\threex) = \propasymn(t'-t,
\threexp-\threex)$, which represents the propagation of an on-shell,
positive-energy particle.  That is, particles propagate into or out of
$\Vol(t\F,t\I)$ on-shell.  Within $\Vol(t\F,t\I)$, $x'$ and $x$ are
not time-ordered, so propagation is still off-shell and virtual, as
given by $\propsymn(x'-x)$.

Thus, rather than constructing virtual-particle propagation from
on-shell propagation, as in traditional QFT, we have recovered
on-shell propagation as a special case.  The basic $\oppsinar(x)$
fields themselves represent off-shell virtual particles and do not in
general satisfy the Klein-Gordon equation.  On the other hand, it is
straightforward to also construct fields that are positive-energy
solutions of the Klein-Gordon equation:
\begin{equation*}
    \adv{\oppsinar}(x) 
	\equiv \intfour \xz\, \propan \oppsinarlz{\xz} \,,
\end{equation*}
Such fields directly represent on-shell particles and antiparticles,
such that
\begin{equation*}
    [\adv{\oppsinar}(x'), \oppsinar\dadj(x)] = \propasymn(x'-x) \,.
\end{equation*}
Now, suppose that, in each term in the expansion of
$\opG_{\Vol(t\I,t\F)}$, we replace all occurrences of $\oppsinar(x)$
with $\adv{\oppsinar}(x)$ (but do \emph{not} replace the
$\oppsinar\dadj(x)$).  Call the resulting operator $\opS(t\F,t\I)$.

To determine the effect of this replacement, consider that the vertex
operator $\opV_{\Vol(t\F,t\I)}$ may be written
\begin{equation*}
    \opV_{\Vol(t\F,t\I)} = \int_{t\I}^{t\F} \dif t \intthree x\, 
	\opV(t,\threex) \,.
\end{equation*}
Because the $\opV(t,\threex)$ commute for different $t$ (per
\thm{thm:vertex}), each term in the series expansion of
$\opG_{\Vol(t\F,t\I)}$ may then be rewritten to time-order the vertex
operators:
\begin{equation*}
    \begin{split}
	\frac{1}{m!}\opV_{\Vol(t\I,t\F)}^{m}
	    &= \int_{t\I}^{t\F} \dt_{m} \int_{t\I}^{t_{m}} \dt_{m-1}
		    \cdots \int_{t\I}^{t_{2}} \dt_{1} \,
		    \opV(t_{m}) \cdots \opV(t_{1}) \\
	    &= \frac{1}{m!}
	       \int_{t\I}^{t\F} \dt_{m} \int_{t\I}^{t\F} \dt_{m-1}
		    \cdots \int_{t\I}^{t\F} \dt_{1} \,
		    T[\opV(t_{m}) \cdots \opV(t_{1})] \,,
    \end{split}
\end{equation*}
where
\begin{equation*}
    \opV(t) \equiv \intthree x\, \opV(t,\threex)
\end{equation*}
and $T[\cdots]$ indicates the sum of all time-ordered permutations of
the bracketed factors.  

Make the $\oppsinar$ to $\adv{\oppsinar}$ replacement in $\opV(t)$ and
call the resulting operator $\opv(t)$.  The series expansion for
$\opS(t\F,t\I)$ is then
\begin{equation*}
    \opS(t\F, t\I)
	\equiv \sum_{m} \frac{(-\mi)^{m}}{m!}
		\int_{t\I}^{t\F} \dt_{m} \int_{t\I}^{t\F} \dt_{m-1}
		    \cdots \int_{t\I}^{t\F} \dt_{1} \,
		    T[\opv(t_{m}) \cdots \opv(t_{1})]
	=      \me^{\T[-\mi\int_{t\I}^{t\F} \dt\, \opv(t)]} \,,
\end{equation*}
where time ordering is now not optional, since the $\opv(t)$ do not
commute. This is essentially just a Dyson expansion.

Make similar field operator replacements in $\bra{\Phi\F}$ to get
$\bra{\adv{\Phi}\F}$.  Then the amplitude $\bra{\adv{\Phi}\F}
\opS(t\F,t\I) \ket{\Phi\I}$ has an expansion that is term-for-term
parallel to the expansion of $\bra{\Phi\F} \opG_{\Vol{t\F,t\I}
}\ket{\Phi\I}$, but with the appropriate operator replacements in each
term.  For incoming and outgoing particles, this will clearly result
in commutators that generate the same $\propasym(x'-x)$ factors as
before the replacement.  For internal propagations, the result will be
as in the usual derivation for Feynman diagrams from the time-ordered
Dyson series \cite{peskin95,weinberg95,ticciati99}:
\begin{equation*}
    \begin{split}
	&\theta(t'-t)[\adv{\oppsina}(t',\threexp),
	              \oppsina\dadj(t,\threex)] +
	 \theta(t-t')[\adv{\oppsinr}(t,\threex), 
		     \oppsinr\dadj(t',\threexp)] \\
	&\qquad\qquad\qquad\qquad = 
	    \theta(t'-t)\propasymn(t'-t,\threexp-\threex) +
	    \theta(t-t')\propasymn(t-t',\threex-\threexp) \\
	&\qquad\qquad\qquad\qquad = 
	    \theta(t'-t)\propasymn(t'-t,\threexp-\threex) +
	    \theta(t-t')\propasymn(t'-t,\threexp-\threex)\conj \\
	&\qquad\qquad\qquad\qquad = 
	    \propsymn(t'-t, \threexp-\threex) \,,
    \end{split}
\end{equation*}
which is virtual particle propagation, as before.  The expansions
before and after the field replacements therefore produce the same
propagation factors, so
\begin{equation*}
    \bra{\adv{\Phi}\F}\opS(t\F,t\I)\ket{\Phi\I} =
	\bra{\Phi\F}\opG_{\Vol(t\F,t\I)}\ket{\Phi\I} \,.
\end{equation*}

Now let $t\F \rightarrow +\infty$ and $t\I \rightarrow -\infty$.  Then
$\opS(t\F,t\I) \rightarrow \opS$, the traditional scattering operator,
and $\Vol(t\F,t\I)$ goes to all spacetime, so $\opG_{\Vol(t\F,t\I)}
\rightarrow \opG$.  In taking the time limits in $\bra{\Phi\F}$ and
$\ket{\Phi\I}$, we would like to hold the three-momenta of the
particles to fixed, given values.  To do this, construct these states
from the time-dependent three-momentum field operators
\begin{equation*}
    \oppsinar(t,\threep) \equiv (2\pi)^{-3/2} \intthree x\,
	\me^{-\mi(-\Ep t + \threep\cdot\threex)}
	\oppsinar(t,\threex)
\end{equation*}
with their respective adjoints.  The corresponding three-momentum
field operators to use in $\bra{\adv{\Phi}\F}$ are
\begin{equation*}
    \adv{\oppsinar}(t,\threep) = (2\pi)^{-3/2} \intthree x\,
	\me^{-\mi(-\Ep t + \threep\cdot\threex)}
	\adv{\oppsinar}(t,\threex) \,.
\end{equation*}

For an incoming particle, the propagation factor in
$\bra{\adv{\Phi}\F} \opS \ket{\Phi\I}$ is then given by
\begin{equation*}
    \begin{split}
	[\adv{\oppsinar}(t', \threexp), \oppsinar\dadj(t, \threep)]
	    &= (2\pi)^{-3/2} \intthree x\,
		\me^{\mi(-\Ep t + \threep\cdot\threex)}
		\propasymn(t'-t, \threexp-\threex) \\
	    &= (2\pi)^{-3/2}(2\Ep)^{-1}
		\me^{\mi(-\Ep t' + \threep \cdot \threexp)} \,.
    \end{split}
\end{equation*}
This is the proper factor for an incoming, on-shell particle of
three-momentum $\threep$ and is independent of $t$ as $t \rightarrow
-\infty$. For an outgoing particle, the factor is
\begin{equation*}
    \begin{split}
	[\adv{\oppsinar}(t', \threepp), \oppsinar\dadj(t, \threex)]
	    &= (2\pi)^{-3/2} \intthree x'\,
		\me^{-\mi(-\Ep t' + \threepp\cdot\threexp)}
		\propasymn(t'-t, \threexp-\threex) \\
	    &= (2\pi)^{-3/2}(2\Ep)^{-1}
		\me^{-\mi(-\Epp t + \threepp \cdot \threex)} \,.
    \end{split}
\end{equation*}
Again, this is the proper factor for an outgoing, on-shell particle 
of three-momentum $\threepp$ and is independent of $t'$ as $t' 
\rightarrow +\infty$.

Therefore, in the $t\F \rightarrow +\infty$, $t\I \rightarrow -\infty$
limit taken as above, $\bra{\adv{\Phi\F}} \opS \ket{\Phi\I}$ is just
the traditional scattering amplitude between on-shell, three-momentum
states.  And, since, in this limit, $\bra{\adv{\Phi}\F}\opS\ket{\Phi\I}
= \bra{\Phi\F}\opG\ket{\Phi\I}$, $\bra{\Phi\F}\opG\ket{\Phi\I}$ is the
same scattering amplitude.

\section{Conclusion} \label{sect:conclusion}

Mathematically, the parameterized approach is appealing.  It makes the
formalism for relativistic quantum theory much more parallel to that
of the non-relativistic theory, and it does not treat time in a
privileged way that is in conflict with the nature of relativistic
spacetime.  Nevertheless, it is reasonable to ask whether this benefit
is worth the cost of introducing another parameter that is (at least
as presented here), in the end, physically unobservable.

Recall, though, that all the empirically tested results of traditional
QFT come from canonical, perturbative QFT. And, even before issues of
regularization and renormalization are addressed, Haag's theorem
already renders the formalism of perturbative QFT inconsistent.  To
date, it has been impossible to develop any formally rigorous,
axiomatic theory of QFT with interactions.

The parameterized formulation presented here resolves this problem.
The key to doing so is basing the theory on a Hilbert space of
off-shell states, as given in \sect{sect:hilbert}.  A Hamiltonian
operator on such states is then defined as a generator of propagation
in the path parameter, rather than time.  And a vacuum state is a null
eigenstate unique relative to a chosen Hamiltonian operator, rather
than being \emph{a priori} unique in the Hilbert space.

The corresponding axioms for fields given in \sect{sect:fields} admit
a free-field theory equivalent to that of traditional QFT, as shown in
\sect{sect:free}.  However, unlike traditional, canonical QFT, Haag's
theorem is absent from parameterized QFT, crucially because the
Hamiltonian generates evolution in a parameter separate from the
spacetime coordinates that are constrained by Poincar\'e invariance.
It is then relatively straightforward to construct an interacting
theory that also meets the parameterized QFT axioms, as shown in
\sect{sect:interacting}.

Further, by allowing the Fock representation of a free field to be
extended to the corresponding interacting field (as discussed in
\sect{sect:states}), the intuitive quanta interpretation of the free
theory can be carried over to the interacting theory.  Indeed, the
approach can also provide for a fuller interpretation in terms of
spacetime paths, decoherence and consistent histories over spacetime,
as addressed in \refcites{seidewitz06b,seidewitz11}.

And, as shown in \sect{sect:scattering}, the parameterized formalism
can be used to derive series expansions for scattering amplitudes that
formally match those derived using traditional perturbative QFT, term
for term.  Admittedly, this shows no more than that the key
empirically tested results of traditional QFT can be duplicated by
parameterized QFT. But it does explain the otherwise surprising fact
that perturbative QFT produces excellent empirical results even
though, in the face of Haag's theorem, it is mathematically
inconsistent as usually formulated.

Moreover, as argued in \refcite{fraser06}, different formulations of
QFT may lead to different interpretations, even while being
empirically equivalent.  Clearly, one would like to base any
interpretation on a formulation that is rigorously defined
mathematically.  But this is problematic for traditional canonical
QFT, since models of realistic interactions using the canonical
formulation run afoul of Haag's Theorem.  Parameterized QFT does not
have this problem.

Of course, this does not resolve all the mathematical issues with
traditional QFT, such as those involved in renormalization.  And the
axiomatic approach discussed here has not yet been extended to cover
gauge theories, as required to model real interactions.  Nevertheless,
such issues all revolve around the mathematical treatment of
interactions in QFT, and having a consistent foundation for modeling
interactions therefore seems to be a prerequisite to resolving them.

In that regard, the work presented here is offered as a step toward
building a firmer foundation, both mathematically and
interpretationally, for QFT in general.

\appendix

\section*{Appendix: Commutivity of the Vertex Operator}

\begin{theorem*}
    Let
    \begin{equation*}
	\opV_{\lambda}(x) = 
	    g :\prod_{i} \oppsinnlll{n_{i}}{\lambda}\dadj(x) 
	       \oppsinnlll{n_{i}}{\lambda}(x)
	       \prod_{j} \oppsinnlll{n_{j}}{\lambda}'(x): \,,
    \end{equation*}
    where
    \begin{equation*}
	 \oppsinnlll{n}{\lambda}(x) = 
	     \oppsinnlll{\na}{\lambda}(x) + 
	     \oppsinnlll{\nr}{\lambda}\dadj(x)
    \end{equation*}
    and
    \begin{equation*}
	\oppsinnlll{n}{\lambda}'(x) = 
	    \oppsinnlll{n}{\lambda}(x) + 
	    \oppsinnlll{n}{\lambda}\dadj(x) \,.
    \end{equation*}
    Then
    \begin{equation*}
	[\opV_{\lambda_{1}}(x_{1}), \opV_{\lambda_{2}}(x_{2})] = 0 \,,
    \end{equation*}
    for all values of the $x_{i}$ and $\lambda_{i}$.
\end{theorem*}
\begin{proof}
    First, consider that
    \begin{equation*}
	\begin{split}
	    [\oppsinnlll{\na}{\lambda_{1}}(x_{1}),
             \oppsinnlll{\na}{\lambda_{2}}\dadj(x_{2})]
	       &= \int_{\lambda_{1}}^{\infty}\dl\,
		      \propsym(x_{1}-x_{2}, \lambda-\lambda_{2}) \\
	       &= \int_{\lambda_{1}-\lambda_{2}}^{\infty}\dl\,
		      \propsym(x_{1} - x_{2}, \lambda)
	\end{split}
    \end{equation*}
    and
    \begin{equation*}
	\begin{split}
	    [\oppsinnlll{\nr}{\lambda_{2}}(x_{2}),
             \oppsinnlll{\nr}{\lambda_{1}}\dadj(x_{1})]
	       &= \int_{-\infty}^{\lambda_{2}}\dl\,
		      \propsym(x_{2}-x_{1}, \lambda_{1}-\lambda) \\
	       &= \int_{\infty}^{\lambda_{1}-\lambda_{2}}\dl\,
		      \propsym(x_{2} - x_{1}, -\lambda) \\
	       &= \int_{\lambda_{1}-\lambda_{2}}^{\infty}\dl\,
		      \propsym(x_{1} - x_{2}, \lambda) \\
	       &= [\oppsinnlll{\na}{\lambda_{1}}(x_{1}),
                   \oppsinnlll{\na}{\lambda_{2}}\dadj(x_{2})] \,,
	\end{split}
    \end{equation*}
    so
    \begin{equation*}
	\begin{split}
	    [\oppsinnlll{n}{\lambda_{1}}(x_{1}),
             \oppsinnlll{n}{\lambda_{2}}\dadj(x_{2})]
		&= [\oppsinnlll{\na}{\lambda_{1}}(x_{1}) + 
	            \oppsinnlll{\nr}{\lambda_{1}}\dadj(x_{1}),
		    \oppsinnlll{\na}{\lambda_{2}}\dadj(x_{2}) + 
		    \oppsinnlll{\nr}{\lambda_{2}}(x_{2})]  \\
		&= [\oppsinnlll{\na}{\lambda_{1}}(x_{1}),
		    \oppsinnlll{\na}{\lambda_{2}}\dadj(x_{2})] -
		   [\oppsinnlll{\nr}{\lambda_{2}}(x_{2}),
		    \oppsinnlll{\nr}{\lambda_{1}}\dadj(x_{1})] \\
		&= 0 \,,
	\end{split}
    \end{equation*}
    and, similarly, $[\oppsinnlll{n}{\lambda_{1}}\dadj(x_{1}),
    \oppsinnlll{n}{\lambda_{2}}(x_{2})] = 0$.  Next, formally
    defined the normal-ordered products by the limiting process
    \begin{equation*}
	\begin{split}
	    :\oppsinnlll{n}{\lambda_{i}}\dadj(x_{i})
	     \oppsinnlll{n}{\lambda_{i}}(x_{i}): \,
		&= \lim_{x'_{i}, x''_{i} \rightarrow x_{i}}
		   \{\oppsinnlll{n}{\lambda_{i}}\dadj(x'_{i})
		     \oppsinnlll{n}{\lambda_{i}}(x''_{i}) -
		     [\oppsinnlll{\nr}{\lambda_{i}}(x''_{i}),
		      \oppsinnlll{\nr}{\lambda_{i}}\dadj(x'_{i})]\} \\
		&= \lim_{x'_{i}, x''_{i} \rightarrow x_{i}}
		   \{\oppsinnlll{n}{\lambda_{i}}\dadj(x'_{i})
		     \oppsinnlll{n}{\lambda_{i}}(x''_{i}) - 
		     \propsym(x''_{i} - x'_{i})\} \,.
	\end{split}
    \end{equation*}
    Then
    \begin{equation*}
	\begin{split}
	    [:\oppsinnlll{n}{\lambda_{1}}\dadj(x_{1})
		\oppsinnlll{n}{\lambda_{1}}(x_{1}):,
		:\oppsinnlll{n}{\lambda_{2}}\dadj(x_{2})
		\oppsinnlll{n}{\lambda_{2}}(x_{2}):]
		& = \lim_{\substack{
			 x'_{1}, x''_{1} \rightarrow x_{1}\\ 
			 x'_{2}, x''_{2} \rightarrow x_{2}}}
		 [\oppsinnlll{n}{\lambda_{1}}\dadj(x_{1}')
		    \oppsinnlll{n}{\lambda_{1}}(x_{1}''),
		    \oppsinnlll{n}{\lambda_{2}}\dadj(x_{2}')
		    \oppsinnlll{n}{\lambda_{2}}(x_{2}'')] \\
		& = 0 \,.
 	\end{split}
    \end{equation*}
    Further,
    \begin{equation*}
	\begin{split}
	    [\oppsinnlll{n}{\lambda_{1}}'(x_{1}),
	     \oppsinnlll{n}{\lambda_{2}}'(x_{2})]
		&= [\oppsinnlll{n}{\lambda_{1}}(x_{1}) + 
		    \oppsinnlll{n}{\lambda_{1}}\dadj(x_{1}),
		    \oppsinnlll{n}{\lambda_{2}}\dadj(x_{2}) + 
		    \oppsinnlll{n}{\lambda_{2}}(x_{2})] \\
		&= [\oppsinnlll{n}{\lambda_{1}}(x_{1}),
		    \oppsinnlll{n}{\lambda_{2}}\dadj(x_{2})] +
		   [\oppsinnlll{n}{\lambda_{1}}\dadj(x_{1}),
		    \oppsinnlll{n}{\lambda_{2}}(x_{2})]
		 = 0 \,.
	\end{split}
    \end{equation*}
    
    Thus, the factor for each particle type in the
    definition of $\opV_{\lambda_{1}}(x_{1})$ commutes with the
    similar factor in the definition of $\opV_{\lambda_{2}}(x_{2})$.
    Since fields of different particle types all commute with each
    other, $\opV_{\lambda_{1}}(x_{1})$ as a whole commutes with
    $\opV_{\lambda_{2}}(x_{2})$ as a whole.
\end{proof}

\bibliography{../../RQMbib}

\begin{thebibliography}{58}%
\makeatletter
\providecommand \@ifxundefined [1]{%
 \@ifx{#1\undefined}
}%
\providecommand \@ifnum [1]{%
 \ifnum #1\expandafter \@firstoftwo
 \else \expandafter \@secondoftwo
 \fi
}%
\providecommand \@ifx [1]{%
 \ifx #1\expandafter \@firstoftwo
 \else \expandafter \@secondoftwo
 \fi
}%
\providecommand \natexlab [1]{#1}%
\providecommand \enquote  [1]{``#1''}%
\providecommand \bibnamefont  [1]{#1}%
\providecommand \bibfnamefont [1]{#1}%
\providecommand \citenamefont [1]{#1}%
\providecommand \href@noop [0]{\@secondoftwo}%
\providecommand \href [0]{\begingroup \@sanitize@url \@href}%
\providecommand \@href[1]{\@@startlink{#1}\@@href}%
\providecommand \@@href[1]{\endgroup#1\@@endlink}%
\providecommand \@sanitize@url [0]{\catcode `\\12\catcode `\$12\catcode
  `\&12\catcode `\#12\catcode `\^12\catcode `\_12\catcode `\%12\relax}%
\providecommand \@@startlink[1]{}%
\providecommand \@@endlink[0]{}%
\providecommand \url  [0]{\begingroup\@sanitize@url \@url }%
\providecommand \@url [1]{\endgroup\@href {#1}{\urlprefix }}%
\providecommand \urlprefix  [0]{URL }%
\providecommand \Eprint [0]{\href }%
\providecommand \doibase [0]{http://dx.doi.org/}%
\providecommand \selectlanguage [0]{\@gobble}%
\providecommand \bibinfo  [0]{\@secondoftwo}%
\providecommand \bibfield  [0]{\@secondoftwo}%
\providecommand \translation [1]{[#1]}%
\providecommand \BibitemOpen [0]{}%
\providecommand \bibitemStop [0]{}%
\providecommand \bibitemNoStop [0]{.\EOS\space}%
\providecommand \EOS [0]{\spacefactor3000\relax}%
\providecommand \BibitemShut  [1]{\csname bibitem#1\endcsname}%
\let\auto@bib@innerbib\@empty
\bibitem [{\citenamefont {Peskin}\ and\ \citenamefont
  {Schroeder}(1995)}]{peskin95}%
  \BibitemOpen
  \bibfield  {author} {\bibinfo {author} {\bibfnamefont {M.~E.}\ \bibnamefont
  {Peskin}}\ and\ \bibinfo {author} {\bibfnamefont {D.~V.}\ \bibnamefont
  {Schroeder}},\ }\href@noop {} {\emph {\bibinfo {title} {An Introduction to
  Quantum Field Theory}}}\ (\bibinfo  {publisher} {Addison-Wesley},\ \bibinfo
  {address} {Reading, Massachusetts},\ \bibinfo {year} {1995})\BibitemShut
  {NoStop}%
\bibitem [{\citenamefont {Weinberg}(1995)}]{weinberg95}%
  \BibitemOpen
  \bibfield  {author} {\bibinfo {author} {\bibfnamefont {S.}~\bibnamefont
  {Weinberg}},\ }\href@noop {} {\emph {\bibinfo {title} {The Quantum Theory of
  Fields}}},\ Vol.\ \bibinfo {volume} {1. \emph{Foundations}}\ (\bibinfo
  {publisher} {Cambridge University Press},\ \bibinfo {address} {Cambridge},\
  \bibinfo {year} {1995})\BibitemShut {NoStop}%
\bibitem [{\citenamefont {Feynman}(1961)}]{feynman61}%
  \BibitemOpen
  \bibfield  {author} {\bibinfo {author} {\bibfnamefont {R.~P.}\ \bibnamefont
  {Feynman}},\ }\href@noop {} {\emph {\bibinfo {title} {Theory of Fundamental
  Processes}}}\ (\bibinfo  {publisher} {Addison-Wesley},\ \bibinfo {address}
  {Readnig},\ \bibinfo {year} {1961})\BibitemShut {NoStop}%
\bibitem [{\citenamefont {DeWitt}(1967)}]{dewitt67}%
  \BibitemOpen
  \bibfield  {author} {\bibinfo {author} {\bibfnamefont {B.~S.}\ \bibnamefont
  {DeWitt}},\ }\href@noop {} {\bibfield  {journal} {\bibinfo  {journal} {Phys.\
  Rev.}\ }\textbf {\bibinfo {volume} {160}},\ \bibinfo {pages} {1113} (\bibinfo
  {year} {1967})}\BibitemShut {NoStop}%
\bibitem [{\citenamefont {Hartle}\ and\ \citenamefont
  {Hawking}(1983)}]{hartle83}%
  \BibitemOpen
  \bibfield  {author} {\bibinfo {author} {\bibfnamefont {J.~B.}\ \bibnamefont
  {Hartle}}\ and\ \bibinfo {author} {\bibfnamefont {S.~W.}\ \bibnamefont
  {Hawking}},\ }\href@noop {} {\bibfield  {journal} {\bibinfo  {journal}
  {Phys.\ Rev.\ D}\ }\textbf {\bibinfo {volume} {28}},\ \bibinfo {pages} {2960}
  (\bibinfo {year} {1983})}\BibitemShut {NoStop}%
\bibitem [{\citenamefont {Dirac}(1950)}]{dirac50}%
  \BibitemOpen
  \bibfield  {author} {\bibinfo {author} {\bibfnamefont {P.~A.~M.}\
  \bibnamefont {Dirac}},\ }\href@noop {} {\bibfield  {journal} {\bibinfo
  {journal} {Can.\ J.\ Math.\ Phys.}\ }\textbf {\bibinfo {volume} {2}},\
  \bibinfo {pages} {129} (\bibinfo {year} {1950})}\BibitemShut {NoStop}%
\bibitem [{\citenamefont {Dirac}(1964)}]{dirac64}%
  \BibitemOpen
  \bibfield  {author} {\bibinfo {author} {\bibfnamefont {P.~A.~M.}\
  \bibnamefont {Dirac}},\ }\href@noop {} {\emph {\bibinfo {title} {Lectures on
  Quantum Mechanics}}}\ (\bibinfo  {publisher} {Belfer Graduate School of
  Science, Yeshiva University},\ \bibinfo {address} {NewY ork},\ \bibinfo
  {year} {1964})\BibitemShut {NoStop}%
\bibitem [{\citenamefont {Rovelli}(2004)}]{rovelli04}%
  \BibitemOpen
  \bibfield  {author} {\bibinfo {author} {\bibfnamefont {C.}~\bibnamefont
  {Rovelli}},\ }\href@noop {} {\emph {\bibinfo {title} {Quantum Gravity}}}\
  (\bibinfo  {publisher} {Cambridge University Press},\ \bibinfo {address}
  {Cambridge},\ \bibinfo {year} {2004})\BibitemShut {NoStop}%
\bibitem [{\citenamefont {Stueckelberg}(1941)}]{stueckelberg41}%
  \BibitemOpen
  \bibfield  {author} {\bibinfo {author} {\bibfnamefont {E.~C.~G.}\
  \bibnamefont {Stueckelberg}},\ }\href@noop {} {\bibfield  {journal} {\bibinfo
   {journal} {Helv.\ Phys.\ Acta}\ }\textbf {\bibinfo {volume} {14}},\ \bibinfo
  {pages} {588} (\bibinfo {year} {1941})}\BibitemShut {NoStop}%
\bibitem [{\citenamefont {Feynman}(1950)}]{feynman50}%
  \BibitemOpen
  \bibfield  {author} {\bibinfo {author} {\bibfnamefont {R.~P.}\ \bibnamefont
  {Feynman}},\ }\href@noop {} {\bibfield  {journal} {\bibinfo  {journal}
  {Phys.\ Rev.}\ }\textbf {\bibinfo {volume} {80}},\ \bibinfo {pages} {440}
  (\bibinfo {year} {1950})}\BibitemShut {NoStop}%
\bibitem [{\citenamefont {Teitelboim}(1982)}]{teitelboim82}%
  \BibitemOpen
  \bibfield  {author} {\bibinfo {author} {\bibfnamefont {C.}~\bibnamefont
  {Teitelboim}},\ }\href@noop {} {\bibfield  {journal} {\bibinfo  {journal}
  {Phys.\ Rev.\ D}\ }\textbf {\bibinfo {volume} {25}},\ \bibinfo {pages} {3159}
  (\bibinfo {year} {1982})}\BibitemShut {NoStop}%
\bibitem [{\citenamefont {Halliwell}\ and\ \citenamefont
  {Thorwart}(2001)}]{halliwell01b}%
  \BibitemOpen
  \bibfield  {author} {\bibinfo {author} {\bibfnamefont {J.~J.}\ \bibnamefont
  {Halliwell}}\ and\ \bibinfo {author} {\bibfnamefont {J.}~\bibnamefont
  {Thorwart}},\ }\href@noop {} {\bibfield  {journal} {\bibinfo  {journal}
  {Phys.\ Rev.\ D}\ }\textbf {\bibinfo {volume} {64}},\ \bibinfo {pages}
  {124018} (\bibinfo {year} {2001})}\BibitemShut {NoStop}%
\bibitem [{\citenamefont {Seidewitz}(2006)}]{seidewitz06a}%
  \BibitemOpen
  \bibfield  {author} {\bibinfo {author} {\bibfnamefont {E.}~\bibnamefont
  {Seidewitz}},\ }\href@noop {} {\bibfield  {journal} {\bibinfo  {journal} {J.
  Math. Phys.}\ }\textbf {\bibinfo {volume} {47}},\ \bibinfo {pages} {112302}
  (\bibinfo {year} {2006})},\ \bibinfo {note}
  {arxiv:quant-ph/0507115}\BibitemShut {NoStop}%
\bibitem [{\citenamefont {Fock}(1937)}]{fock37}%
  \BibitemOpen
  \bibfield  {author} {\bibinfo {author} {\bibfnamefont {V.~A.}\ \bibnamefont
  {Fock}},\ }\href@noop {} {\bibfield  {journal} {\bibinfo  {journal} {Physik
  Z. Sowjetunion}\ }\textbf {\bibinfo {volume} {12}},\ \bibinfo {pages} {404}
  (\bibinfo {year} {1937})}\BibitemShut {NoStop}%
\bibitem [{\citenamefont {Stueckelberg}(1942)}]{stueckelberg42}%
  \BibitemOpen
  \bibfield  {author} {\bibinfo {author} {\bibfnamefont {E.~C.~G.}\
  \bibnamefont {Stueckelberg}},\ }\href@noop {} {\bibfield  {journal} {\bibinfo
   {journal} {Helv.\ Phys.\ Acta}\ }\textbf {\bibinfo {volume} {15}},\ \bibinfo
  {pages} {23} (\bibinfo {year} {1942})}\BibitemShut {NoStop}%
\bibitem [{\citenamefont {Nambu}(1950)}]{nambu50}%
  \BibitemOpen
  \bibfield  {author} {\bibinfo {author} {\bibfnamefont {Y.}~\bibnamefont
  {Nambu}},\ }\href@noop {} {\bibfield  {journal} {\bibinfo  {journal} {Progr.\
  Theoret.\ Phys.}\ }\textbf {\bibinfo {volume} {5}},\ \bibinfo {pages} {82}
  (\bibinfo {year} {1950})}\BibitemShut {NoStop}%
\bibitem [{\citenamefont {Feynman}(1951)}]{feynman51}%
  \BibitemOpen
  \bibfield  {author} {\bibinfo {author} {\bibfnamefont {R.~P.}\ \bibnamefont
  {Feynman}},\ }\href@noop {} {\bibfield  {journal} {\bibinfo  {journal}
  {Phys.\ Rev.}\ }\textbf {\bibinfo {volume} {84}},\ \bibinfo {pages} {108}
  (\bibinfo {year} {1951})}\BibitemShut {NoStop}%
\bibitem [{\citenamefont {Schwinger}(1951)}]{schwinger51}%
  \BibitemOpen
  \bibfield  {author} {\bibinfo {author} {\bibfnamefont {J.}~\bibnamefont
  {Schwinger}},\ }\href@noop {} {\bibfield  {journal} {\bibinfo  {journal}
  {Phys.\ Rev.}\ }\textbf {\bibinfo {volume} {82}},\ \bibinfo {pages} {664}
  (\bibinfo {year} {1951})}\BibitemShut {NoStop}%
\bibitem [{\citenamefont {Morette}(1951)}]{morette51}%
  \BibitemOpen
  \bibfield  {author} {\bibinfo {author} {\bibfnamefont {C.}~\bibnamefont
  {Morette}},\ }\href@noop {} {\bibfield  {journal} {\bibinfo  {journal} {Phys.
  Rev.}\ }\textbf {\bibinfo {volume} {81}},\ \bibinfo {pages} {848} (\bibinfo
  {year} {1951})}\BibitemShut {NoStop}%
\bibitem [{\citenamefont {Cooke}(1968)}]{cooke68}%
  \BibitemOpen
  \bibfield  {author} {\bibinfo {author} {\bibfnamefont {J.~H.}\ \bibnamefont
  {Cooke}},\ }\href@noop {} {\bibfield  {journal} {\bibinfo  {journal} {Phys.\
  Rev.}\ }\textbf {\bibinfo {volume} {166}},\ \bibinfo {pages} {1293} (\bibinfo
  {year} {1968})}\BibitemShut {NoStop}%
\bibitem [{\citenamefont {Horwitz}\ and\ \citenamefont
  {Piron}(1973)}]{horwitz73}%
  \BibitemOpen
  \bibfield  {author} {\bibinfo {author} {\bibfnamefont {L.~P.}\ \bibnamefont
  {Horwitz}}\ and\ \bibinfo {author} {\bibfnamefont {C.}~\bibnamefont
  {Piron}},\ }\href@noop {} {\bibfield  {journal} {\bibinfo  {journal} {Helv.\
  Phys.\ Acta}\ }\textbf {\bibinfo {volume} {46}},\ \bibinfo {pages} {316}
  (\bibinfo {year} {1973})}\BibitemShut {NoStop}%
\bibitem [{\citenamefont {Piron}\ and\ \citenamefont {Reuse}(1978)}]{piron78}%
  \BibitemOpen
  \bibfield  {author} {\bibinfo {author} {\bibfnamefont {C.}~\bibnamefont
  {Piron}}\ and\ \bibinfo {author} {\bibfnamefont {F.}~\bibnamefont {Reuse}},\
  }\href@noop {} {\bibfield  {journal} {\bibinfo  {journal} {Helv.\ Phys.\
  Acta}\ }\textbf {\bibinfo {volume} {51}},\ \bibinfo {pages} {146} (\bibinfo
  {year} {1978})}\BibitemShut {NoStop}%
\bibitem [{\citenamefont {Collins}\ and\ \citenamefont
  {Fanchi}(1978)}]{collins78}%
  \BibitemOpen
  \bibfield  {author} {\bibinfo {author} {\bibfnamefont {R.~E.}\ \bibnamefont
  {Collins}}\ and\ \bibinfo {author} {\bibfnamefont {J.~R.}\ \bibnamefont
  {Fanchi}},\ }\href@noop {} {\bibfield  {journal} {\bibinfo  {journal} {Nuovo
  Cimento}\ }\textbf {\bibinfo {volume} {48A}},\ \bibinfo {pages} {314}
  (\bibinfo {year} {1978})}\BibitemShut {NoStop}%
\bibitem [{\citenamefont {Fanchi}\ and\ \citenamefont
  {Collins}(1978)}]{fanchi78}%
  \BibitemOpen
  \bibfield  {author} {\bibinfo {author} {\bibfnamefont {J.~R.}\ \bibnamefont
  {Fanchi}}\ and\ \bibinfo {author} {\bibfnamefont {R.~E.}\ \bibnamefont
  {Collins}},\ }\href@noop {} {\bibfield  {journal} {\bibinfo  {journal}
  {Found.\ Phys.}\ }\textbf {\bibinfo {volume} {8}},\ \bibinfo {pages} {851}
  (\bibinfo {year} {1978})}\BibitemShut {NoStop}%
\bibitem [{\citenamefont {Fanchi}\ and\ \citenamefont
  {Wilson}(1983)}]{fanchi83}%
  \BibitemOpen
  \bibfield  {author} {\bibinfo {author} {\bibfnamefont {J.~R.}\ \bibnamefont
  {Fanchi}}\ and\ \bibinfo {author} {\bibfnamefont {W.~J.}\ \bibnamefont
  {Wilson}},\ }\href@noop {} {\bibfield  {journal} {\bibinfo  {journal}
  {Found.\ Phys.}\ }\textbf {\bibinfo {volume} {13}},\ \bibinfo {pages} {571}
  (\bibinfo {year} {1983})}\BibitemShut {NoStop}%
\bibitem [{\citenamefont {Fanchi}(1993)}]{fanchi93}%
  \BibitemOpen
  \bibfield  {author} {\bibinfo {author} {\bibfnamefont {J.~R.}\ \bibnamefont
  {Fanchi}},\ }\href@noop {} {\emph {\bibinfo {title} {Parametrized
  Relativistic Quantum Theory}}}\ (\bibinfo  {publisher} {Kluwer Academic},\
  \bibinfo {address} {Dordrecht},\ \bibinfo {year} {1993})\BibitemShut
  {NoStop}%
\bibitem [{\citenamefont {Hartle}\ and\ \citenamefont
  {Kucha\v{r}}(1986)}]{hartle86}%
  \BibitemOpen
  \bibfield  {author} {\bibinfo {author} {\bibfnamefont {J.~B.}\ \bibnamefont
  {Hartle}}\ and\ \bibinfo {author} {\bibfnamefont {K.~V.}\ \bibnamefont
  {Kucha\v{r}}},\ }\href@noop {} {\bibfield  {journal} {\bibinfo  {journal}
  {Phys. Rev. D}\ }\textbf {\bibinfo {volume} {34}},\ \bibinfo {pages} {2323}
  (\bibinfo {year} {1986})}\BibitemShut {NoStop}%
\bibitem [{\citenamefont {Hartle}(1995)}]{hartle95}%
  \BibitemOpen
  \bibfield  {author} {\bibinfo {author} {\bibfnamefont {J.~B.}\ \bibnamefont
  {Hartle}},\ }in\ \href@noop {} {\emph {\bibinfo {booktitle} {Gravitation and
  Quantizations: Proceedings of the 1992 Les Houches Summer School}}},\
  \bibinfo {editor} {edited by\ \bibinfo {editor} {\bibfnamefont
  {B.}~\bibnamefont {Julia}}\ and\ \bibinfo {editor} {\bibfnamefont
  {J.}~\bibnamefont {Zinn-Justin}}}\ (\bibinfo  {publisher} {North-Holland},\
  \bibinfo {address} {Amsterdam},\ \bibinfo {year} {1995})\ \bibinfo {note}
  {gr-qc/9304006}\BibitemShut {NoStop}%
\bibitem [{\citenamefont {Halliwell}(2001)}]{halliwell01a}%
  \BibitemOpen
  \bibfield  {author} {\bibinfo {author} {\bibfnamefont {J.~J.}\ \bibnamefont
  {Halliwell}},\ }\href@noop {} {\bibfield  {journal} {\bibinfo  {journal}
  {Phys.\ Rev.\ D}\ }\textbf {\bibinfo {volume} {64}},\ \bibinfo {pages}
  {044008} (\bibinfo {year} {2001})}\BibitemShut {NoStop}%
\bibitem [{\citenamefont {Halliwell}\ and\ \citenamefont
  {Thorwart}(2002)}]{halliwell02}%
  \BibitemOpen
  \bibfield  {author} {\bibinfo {author} {\bibfnamefont {J.~J.}\ \bibnamefont
  {Halliwell}}\ and\ \bibinfo {author} {\bibfnamefont {J.}~\bibnamefont
  {Thorwart}},\ }\href@noop {} {\bibfield  {journal} {\bibinfo  {journal}
  {Phys.\ Rev.\ D}\ }\textbf {\bibinfo {volume} {65}},\ \bibinfo {pages}
  {104009} (\bibinfo {year} {2002})}\BibitemShut {NoStop}%
\bibitem [{\citenamefont {Fanchi}(1979)}]{fanchi79}%
  \BibitemOpen
  \bibfield  {author} {\bibinfo {author} {\bibfnamefont {J.~R.}\ \bibnamefont
  {Fanchi}},\ }\href@noop {} {\bibfield  {journal} {\bibinfo  {journal} {Phys.\
  Rev.\ D}\ }\textbf {\bibinfo {volume} {20}},\ \bibinfo {pages} {3108}
  (\bibinfo {year} {1979})}\BibitemShut {NoStop}%
\bibitem [{\citenamefont {Saad}, \citenamefont {Horwitz},\ and\ \citenamefont
  {Arshansky}(1989)}]{saad89}%
  \BibitemOpen
  \bibfield  {author} {\bibinfo {author} {\bibfnamefont {D.}~\bibnamefont
  {Saad}}, \bibinfo {author} {\bibfnamefont {L.~P.}\ \bibnamefont {Horwitz}}, \
  and\ \bibinfo {author} {\bibfnamefont {R.~I.}\ \bibnamefont {Arshansky}},\
  }\href@noop {} {\bibfield  {journal} {\bibinfo  {journal} {Found.\ Phys.}\
  }\textbf {\bibinfo {volume} {19}},\ \bibinfo {pages} {1677} (\bibinfo {year}
  {1989})}\BibitemShut {NoStop}%
\bibitem [{\citenamefont {Land}\ and\ \citenamefont {Horwitz}(1991)}]{land91}%
  \BibitemOpen
  \bibfield  {author} {\bibinfo {author} {\bibfnamefont {M.~C.}\ \bibnamefont
  {Land}}\ and\ \bibinfo {author} {\bibfnamefont {L.~P.}\ \bibnamefont
  {Horwitz}},\ }\href@noop {} {\bibfield  {journal} {\bibinfo  {journal}
  {Found. Phys.}\ }\textbf {\bibinfo {volume} {21}},\ \bibinfo {pages} {299}
  (\bibinfo {year} {1991})}\BibitemShut {NoStop}%
\bibitem [{\citenamefont {Shnerb}\ and\ \citenamefont
  {Horwitz}(1993)}]{shnerb93}%
  \BibitemOpen
  \bibfield  {author} {\bibinfo {author} {\bibfnamefont {N.}~\bibnamefont
  {Shnerb}}\ and\ \bibinfo {author} {\bibfnamefont {L.~P.}\ \bibnamefont
  {Horwitz}},\ }\href@noop {} {\bibfield  {journal} {\bibinfo  {journal} {Phys.
  Rev. A}\ }\textbf {\bibinfo {volume} {48}},\ \bibinfo {pages} {4068}
  (\bibinfo {year} {1993})}\BibitemShut {NoStop}%
\bibitem [{\citenamefont {Pav\v{s}i\v{c}}(1991)}]{pavsic91}%
  \BibitemOpen
  \bibfield  {author} {\bibinfo {author} {\bibfnamefont {M.}~\bibnamefont
  {Pav\v{s}i\v{c}}},\ }\href@noop {} {\bibfield  {journal} {\bibinfo  {journal}
  {Nuovo Cimento A}\ }\textbf {\bibinfo {volume} {104}},\ \bibinfo {pages}
  {1337} (\bibinfo {year} {1991})}\BibitemShut {NoStop}%
\bibitem [{\citenamefont {Horwitz}(1998)}]{horwitz98}%
  \BibitemOpen
  \bibfield  {author} {\bibinfo {author} {\bibfnamefont {L.~P.}\ \bibnamefont
  {Horwitz}},\ }\href@noop {} {\enquote {\bibinfo {title} {Second quantization
  of the stueckelberg relativistic quantum theory and associated guage
  fields},}\ } (\bibinfo {year} {1998}),\ \bibinfo {note}
  {arxiv:hep-ph/9804155}\BibitemShut {NoStop}%
\bibitem [{\citenamefont {Pav\v{s}i\v{c}}(1998)}]{pavsic98}%
  \BibitemOpen
  \bibfield  {author} {\bibinfo {author} {\bibfnamefont {M.}~\bibnamefont
  {Pav\v{s}i\v{c}}},\ }\href@noop {} {\bibfield  {journal} {\bibinfo  {journal}
  {Found.\ Phys.}\ }\textbf {\bibinfo {volume} {28}},\ \bibinfo {pages} {1453}
  (\bibinfo {year} {1998})}\BibitemShut {NoStop}%
\bibitem [{\citenamefont {Seidewitz}(2007)}]{seidewitz06b}%
  \BibitemOpen
  \bibfield  {author} {\bibinfo {author} {\bibfnamefont {E.}~\bibnamefont
  {Seidewitz}},\ }\href@noop {} {\bibfield  {journal} {\bibinfo  {journal}
  {Found. Phys.}\ }\textbf {\bibinfo {volume} {37}},\ \bibinfo {pages} {572}
  (\bibinfo {year} {2007})},\ \bibinfo {note}
  {arxiv:quant-ph/0612023}\BibitemShut {NoStop}%
\bibitem [{\citenamefont {Seidewitz}(2009)}]{seidewitz09}%
  \BibitemOpen
  \bibfield  {author} {\bibinfo {author} {\bibfnamefont {E.}~\bibnamefont
  {Seidewitz}},\ }\href@noop {} {\bibfield  {journal} {\bibinfo  {journal}
  {Ann. Phys.}\ }\textbf {\bibinfo {volume} {324}},\ \bibinfo {pages} {309}
  (\bibinfo {year} {2009})},\ \bibinfo {note} {arxiv:0804.3206
  [quant-ph]}\BibitemShut {NoStop}%
\bibitem [{\citenamefont {Seidewitz}(2011)}]{seidewitz11}%
  \BibitemOpen
  \bibfield  {author} {\bibinfo {author} {\bibfnamefont {E.}~\bibnamefont
  {Seidewitz}},\ }\href@noop {} {\bibfield  {journal} {\bibinfo  {journal}
  {Found. Phys.}\ }\textbf {\bibinfo {volume} {41}},\ \bibinfo {pages} {1163}
  (\bibinfo {year} {2011})},\ \bibinfo {note} {arxiv:1002.3917
  [quant-ph]}\BibitemShut {NoStop}%
\bibitem [{\citenamefont {Wightman}(1959)}]{wightman59}%
  \BibitemOpen
  \bibfield  {author} {\bibinfo {author} {\bibfnamefont {A.~S.}\ \bibnamefont
  {Wightman}},\ }in\ \href@noop {} {\emph {\bibinfo {booktitle} {Les
  probl\`emes math\'ematique de la th\'eorie quantique des champs}}}\ (\bibinfo
   {publisher} {Centre National de la Recherche Scientifique},\ \bibinfo
  {address} {Paris},\ \bibinfo {year} {1959})\ pp.\ \bibinfo {pages}
  {11--19}\BibitemShut {NoStop}%
\bibitem [{\citenamefont {Wightman}(1964)}]{wightman64}%
  \BibitemOpen
  \bibfield  {author} {\bibinfo {author} {\bibfnamefont {A.~S.}\ \bibnamefont
  {Wightman}},\ }\href@noop {} {\bibfield  {journal} {\bibinfo  {journal}
  {Ark.\ Fys.}\ }\textbf {\bibinfo {volume} {28}},\ \bibinfo {pages} {129}
  (\bibinfo {year} {1964})}\BibitemShut {NoStop}%
\bibitem [{\citenamefont {Haag}(1993)}]{haag93}%
  \BibitemOpen
  \bibfield  {author} {\bibinfo {author} {\bibfnamefont {R.}~\bibnamefont
  {Haag}},\ }\href@noop {} {\emph {\bibinfo {title} {Local Quantum Physics}}},\
  \bibinfo {edition} {2nd}\ ed.\ (\bibinfo  {publisher} {Springer-Verlag},\
  \bibinfo {address} {Berlin},\ \bibinfo {year} {1993})\BibitemShut {NoStop}%
\bibitem [{\citenamefont {Streater}\ and\ \citenamefont
  {Wightman}(1964)}]{streater64}%
  \BibitemOpen
  \bibfield  {author} {\bibinfo {author} {\bibfnamefont {R.~F.}\ \bibnamefont
  {Streater}}\ and\ \bibinfo {author} {\bibfnamefont {A.~S.}\ \bibnamefont
  {Wightman}},\ }\href@noop {} {\emph {\bibinfo {title} {{PCT}, Spin,
  Statistics, and All That}}}\ (\bibinfo  {publisher} {Princeton University
  Press},\ \bibinfo {address} {Princeton},\ \bibinfo {year} {1964})\BibitemShut
  {NoStop}%
\bibitem [{\citenamefont {Bogolubov}, \citenamefont {Logunov},\ and\
  \citenamefont {Todorov}(1975)}]{bogolubov75}%
  \BibitemOpen
  \bibfield  {author} {\bibinfo {author} {\bibfnamefont {N.~N.}\ \bibnamefont
  {Bogolubov}}, \bibinfo {author} {\bibfnamefont {A.~A.}\ \bibnamefont
  {Logunov}}, \ and\ \bibinfo {author} {\bibfnamefont {I.~T.}\ \bibnamefont
  {Todorov}},\ }\href@noop {} {\emph {\bibinfo {title} {Introduction to
  Axiomatic Quantum Field Theory}}}\ (\bibinfo  {publisher} {W. A. Benjamin,
  Inc.},\ \bibinfo {address} {Reading, MA},\ \bibinfo {year}
  {1975})\BibitemShut {NoStop}%
\bibitem [{\citenamefont {Haag}(1957)}]{haag55}%
  \BibitemOpen
  \bibfield  {author} {\bibinfo {author} {\bibfnamefont {R.}~\bibnamefont
  {Haag}},\ }\href@noop {} {\bibfield  {journal} {\bibinfo  {journal} {Dan.\
  Mat.\ Fys.\ Medd.}\ }\textbf {\bibinfo {volume} {29}},\ \bibinfo {pages} {1}
  (\bibinfo {year} {1957})}\BibitemShut {NoStop}%
\bibitem [{\citenamefont {Hall}\ and\ \citenamefont {Wightman}(1957)}]{hall57}%
  \BibitemOpen
  \bibfield  {author} {\bibinfo {author} {\bibfnamefont {D.}~\bibnamefont
  {Hall}}\ and\ \bibinfo {author} {\bibfnamefont {A.~S.}\ \bibnamefont
  {Wightman}},\ }\href@noop {} {\bibfield  {journal} {\bibinfo  {journal}
  {Mat.\ Fys.\ Medd.\ Dan.\ Vids.\ Selsk.}\ }\textbf {\bibinfo {volume} {35}},\
  \bibinfo {pages} {1} (\bibinfo {year} {1957})}\BibitemShut {NoStop}%
\bibitem [{\citenamefont {Seidewitz}(2016)}]{seidewitz15}%
  \BibitemOpen
  \bibfield  {author} {\bibinfo {author} {\bibfnamefont {E.}~\bibnamefont
  {Seidewitz}},\ }\href@noop {} {\enquote {\bibinfo {title} {Avoiding haag's
  theorem with parameterized quantum field theory},}\ } (\bibinfo {year}
  {2016}),\ \bibinfo {note} {arxiv:1501.05658v3 [hep-th]}\BibitemShut {NoStop}%
\bibitem [{\citenamefont {Gel'fand}\ and\ \citenamefont
  {Kostyuchenko}(1964)}]{gelfand55}%
  \BibitemOpen
  \bibfield  {author} {\bibinfo {author} {\bibfnamefont {I.~M.}\ \bibnamefont
  {Gel'fand}}\ and\ \bibinfo {author} {\bibfnamefont {A.~G.}\ \bibnamefont
  {Kostyuchenko}},\ }\href@noop {} {\bibfield  {journal} {\bibinfo  {journal}
  {Dokl. Akad. Nauk SSSR}\ }\textbf {\bibinfo {volume} {103}},\ \bibinfo
  {pages} {349} (\bibinfo {year} {1964})}\BibitemShut {NoStop}%
\bibitem [{\citenamefont {Gel'fand}\ and\ \citenamefont
  {Vilenkin}(1964)}]{gelfand64}%
  \BibitemOpen
  \bibfield  {author} {\bibinfo {author} {\bibfnamefont {I.~M.}\ \bibnamefont
  {Gel'fand}}\ and\ \bibinfo {author} {\bibfnamefont {N.~Y.}\ \bibnamefont
  {Vilenkin}},\ }\href@noop {} {\emph {\bibinfo {title} {Generalized
  Functions}}},\ \bibinfo {series} {Applications of Harmonic Analysis},
  Vol.~\bibinfo {volume} {4}\ (\bibinfo  {publisher} {Academic Press},\
  \bibinfo {address} {New York},\ \bibinfo {year} {1964})\BibitemShut {NoStop}%
\bibitem [{\citenamefont {Pelc}\ and\ \citenamefont {Horwitz}(1997)}]{pelc97}%
  \BibitemOpen
  \bibfield  {author} {\bibinfo {author} {\bibfnamefont {O.}~\bibnamefont
  {Pelc}}\ and\ \bibinfo {author} {\bibfnamefont {L.~P.}\ \bibnamefont
  {Horwitz}},\ }\href@noop {} {\bibfield  {journal} {\bibinfo  {journal} {J.\
  Math.\ Phys}\ }\textbf {\bibinfo {volume} {38}},\ \bibinfo {pages} {115}
  (\bibinfo {year} {1997})},\ \bibinfo {note}
  {arXiv:hep-th/9605148}\BibitemShut {NoStop}%
\bibitem [{\citenamefont {Feynman}(1949)}]{feynman49}%
  \BibitemOpen
  \bibfield  {author} {\bibinfo {author} {\bibfnamefont {R.~P.}\ \bibnamefont
  {Feynman}},\ }\href@noop {} {\bibfield  {journal} {\bibinfo  {journal}
  {Phys.\ Rev.}\ }\textbf {\bibinfo {volume} {76}},\ \bibinfo {pages} {749}
  (\bibinfo {year} {1949})}\BibitemShut {NoStop}%
\bibitem [{\citenamefont {Pauli}\ and\ \citenamefont
  {Villars}(1949)}]{pauli49}%
  \BibitemOpen
  \bibfield  {author} {\bibinfo {author} {\bibfnamefont {W.}~\bibnamefont
  {Pauli}}\ and\ \bibinfo {author} {\bibfnamefont {F.}~\bibnamefont
  {Villars}},\ }\href@noop {} {\bibfield  {journal} {\bibinfo  {journal} {Rev.
  Mod. Phys.}\ }\textbf {\bibinfo {volume} {21}},\ \bibinfo {pages} {434}
  (\bibinfo {year} {1949})}\BibitemShut {NoStop}%
\bibitem [{\citenamefont {Frastai}\ and\ \citenamefont
  {Horwitz}(1995)}]{frastai95}%
  \BibitemOpen
  \bibfield  {author} {\bibinfo {author} {\bibfnamefont {J.}~\bibnamefont
  {Frastai}}\ and\ \bibinfo {author} {\bibfnamefont {L.~P.}\ \bibnamefont
  {Horwitz}},\ }\href@noop {} {\bibfield  {journal} {\bibinfo  {journal}
  {Found. Phys.}\ }\textbf {\bibinfo {volume} {25}},\ \bibinfo {pages} {1485}
  (\bibinfo {year} {1995})}\BibitemShut {NoStop}%
\bibitem [{\citenamefont {Land}(1997)}]{land97}%
  \BibitemOpen
  \bibfield  {author} {\bibinfo {author} {\bibfnamefont {M.~C.}\ \bibnamefont
  {Land}},\ }\href@noop {} {\bibfield  {journal} {\bibinfo  {journal} {Found.
  Phys.}\ }\textbf {\bibinfo {volume} {27}},\ \bibinfo {pages} {19} (\bibinfo
  {year} {1997})},\ \bibinfo {note} {arxiv:hep-th/9701159}\BibitemShut
  {NoStop}%
\bibitem [{\citenamefont {Land}(2003)}]{land03}%
  \BibitemOpen
  \bibfield  {author} {\bibinfo {author} {\bibfnamefont {M.~C.}\ \bibnamefont
  {Land}},\ }\href@noop {} {\bibfield  {journal} {\bibinfo  {journal} {Found.
  Phys.}\ }\textbf {\bibinfo {volume} {33}},\ \bibinfo {pages} {1157} (\bibinfo
  {year} {2003})},\ \bibinfo {note} {arxiv:hep-th/0603074}\BibitemShut
  {NoStop}%
\bibitem [{\citenamefont {Ticciati}(1999)}]{ticciati99}%
  \BibitemOpen
  \bibfield  {author} {\bibinfo {author} {\bibfnamefont {R.}~\bibnamefont
  {Ticciati}},\ }\href@noop {} {\emph {\bibinfo {title} {Quantum Field Theory
  for Mathematicians}}}\ (\bibinfo  {publisher} {Cambridge University Press},\
  \bibinfo {address} {Cambridge},\ \bibinfo {year} {1999})\BibitemShut
  {NoStop}%
\bibitem [{\citenamefont {Fraser}(2006)}]{fraser06}%
  \BibitemOpen
  \bibfield  {author} {\bibinfo {author} {\bibfnamefont {D.~L.}\ \bibnamefont
  {Fraser}},\ }\emph {\bibinfo {title} {Haag's Theorem and the Interpretation
  of Quantum Field Theories with Interactions}},\ \href@noop {} {Ph.D.
  thesis},\ \bibinfo  {school} {University of Pittsburgh} (\bibinfo {year}
  {2006}),\ \bibinfo {note} {http://d-scholarship.pitt.edu
  /8260/1/D\_Fraser\_2006.pdf}\BibitemShut {NoStop}%
\end{thebibliography}%

\end{document}